\documentclass[acmsmall,screen,nonacm]{acmart}

\usepackage{amsmath, bm}
\usepackage{amsfonts}
\usepackage{mathtools}
\usepackage{todonotes}
\usepackage{graphicx}
\usepackage{algorithm}
\usepackage{algpseudocode}
\usepackage{booktabs}
\usepackage{nicefrac}
\usepackage{siunitx}
\usepackage{xspace}
\usepackage{microtype}
\usepackage{colortbl}
\usepackage{parskip}

\usepackage{caption}
\usepackage{subcaption}
\usepackage[shortlabels]{enumitem}

\usepackage{syntax}
\usepackage{pgfplots}
\pgfplotsset{compat=newest}

\usepackage{mleftright}
\usepackage{tikz}
\usetikzlibrary{arrows,automata,calc,positioning,fit}

\usepackage{listings}
\usepackage{lstautogobble}

\usepackage{multirow}
\usepackage{hhline}

\usepackage{pifont}
\usepackage{marvosym}

\newcommand{\xmark}{\ding{55}}

\newcommand{\prog}{\ensuremath{\mathcal{P}}\xspace}
\newcommand{\Scal}{\ensuremath{\mathcal{S}}\xspace}
\newcommand{\K}{\ensuremath{\mathbb{K}}\xspace}
\newcommand{\N}{\ensuremath{\mathbb{N}}\xspace}

\newcommand{\R}{\ensuremath{\mathbb{R}}\xspace}
\newcommand{\Q}{\ensuremath{\mathbb{Q}}\xspace}
\newcommand{\Qbar}{\ensuremath{\overline{\mathbb{Q}}}\xspace}
\newcommand{\E}{\ensuremath{\mathbb{E}}}

\renewcommand{\P}{\ensuremath{\mathbb{P}}}
\newcommand{\I}{\ensuremath{\mathbb{I}}}

\newcommand{\skolem}{\hyperref[skolem]{\textsc{Skolem}}\xspace}
\newcommand{\spinv}{\hyperref[spinv]{\textsc{SPInv}}\xspace}
\newcommand{\probspinv}{\hyperref[probspinv]{\textsc{Prob-SPInv}}\xspace}
\newcommand{\ptp}{\hyperref[ptp]{\textsc{P2P}}\xspace}
\newcommand{\lrs}{LRS\xspace}

\newcommand{\nicehat}[1]{\expandafter\hat#1}

\theoremstyle{acmdefinition}
\newtheorem{remark}{Remark}

\newcommand{\problembox}[1]{
    \medskip
    \noindent\ignorespaces\fbox{%
        \hspace{0.3em}%
        \begin{minipage}{\linewidth - 2\fboxsep - 0.6em}
            {#1}
        \end{minipage}%
        \hspace{0.3em}%
    }%
    \medskip
}

\newenvironment{instance}{%
\medskip
\refstepcounter{equation}
\noindent\ignorespaces%
\begin{minipage}{0.85\linewidth}
}{%
\end{minipage}\hfill
\begin{minipage}{0.1\linewidth}
\hfill(\theequation)
\end{minipage}
\medskip
}  

\let\originalleft\left
\let\originalright\right
\renewcommand{\left}{\mathopen{}\mathclose\bgroup\originalleft}
\renewcommand{\right}{\aftergroup\egroup\originalright}

\algnewcommand{\LeftComment}[1]{\Statex \(\triangleright\) \textit{#1}}

\renewcommand{\paragraph}[1]{\par\noindent{{\bf #1}}}

\begin{document}
\title{Strong Invariants Are Hard}
\subtitle{On the Hardness of Strongest Polynomial Invariants for (Probabilistic) Programs}
\author{Julian M{\"u}llner}
\orcid{0009-0006-2909-2297}
\affiliation{%
  \institution{TU Wien}
  \city{Vienna}
  \country{Austria}
}
\email{julian.muellner@tuwien.ac.at}
\author{Marcel Moosbrugger}
\orcid{0000-0002-2006-3741}
\affiliation{%
  \institution{TU Wien}
  \city{Vienna}
  \country{Austria}
}
\email{marcel.moosbrugger@tuwien.ac.at}
\author{Laura Kov\'acs}
\orcid{0000-0002-8299-2714}
\affiliation{%
  \institution{TU Wien}
  \city{Vienna}
  \country{Austria}
}
\email{laura.kovacs@tuwien.ac.at}

\begin{abstract}
We show that computing the strongest polynomial invariant for single-path loops with polynomial assignments is at least as hard as the \textsc{Skolem} problem, a famous problem whose decidability has been open for almost a century.
While the strongest polynomial invariants are computable for \emph{affine loops}, for polynomial loops the problem remained wide open.
As an intermediate result of independent interest, we prove that reachability for discrete polynomial dynamical systems is \textsc{Skolem}-hard as well.
Furthermore, we generalize the notion of invariant ideals and introduce \emph{moment invariant ideals} for probabilistic programs.
With this tool, we further show that the strongest polynomial moment invariant is (i) uncomputable, for probabilistic loops with branching statements, and (ii) \textsc{Skolem}-hard to compute for polynomial probabilistic loops without branching statements.
Finally, we identify a class of probabilistic loops for which the strongest polynomial moment invariant is computable and provide an algorithm for it.
\end{abstract}



\maketitle

\section{Introduction} \label{sec:intro}

Loop invariants describe valid program properties that hold before and after every loop iteration. Intuitively, invariants provide correctness information that may prevent programmers from introducing errors while making changes to the loop.
As such, invariants are fundamental to formalizing program semantics as well as to automate the formal analysis and verification of programs. 
While automatically synthesizing loop invariants is, in general, an uncomputable problem, when considering only single-path loops with linear updates (linear loops), the strongest polynomial invariant is in fact computable~\cite{Karr76,Muller-OlmS04-2,Kovacs08}.
The computability remains intact for linear loops with non-deterministic branching~\cite{DBLP:conf/lics/HrushovskiOP018}.
Yet, already for single-path loops with \lq\lq only\rq\rq\ polynomial updates, computing the strongest invariant has been an open challenge since 2004~\cite{Muller-OlmS04}.
In this paper, we bridge the gap between the computability result for linear loops and the uncomputability result for general loops by providing, to the best of our knowledge, the \emph{first hardness result for computing the strongest polynomial invariant of single-path polynomial loops}.

\noindent
\begin{figure}[t]
    \begin{subfigure}[b]{0.45\textwidth}
        \centering
        \begin{algorithmic}
        \State{$
            \begin{bmatrix}
                f~ & u~ & v~ & w
            \end{bmatrix}
            \gets
            \begin{bmatrix}
                1~ & -1~ & 2~ & 0
            \end{bmatrix}
            $
        }
        \While{$\star$}
            \State{$t \gets 3t+2u-5w$}
            \State{$u \gets u+3w$}
            \State{$v \gets 4u+3v+w$}
            \State{$w \gets t+u+2v$}
        \EndWhile
        \end{algorithmic}
        \caption{An \emph{affine} loop from \cite{Karimov22}.}
        \label{fig:intro:examples:affine}
    \end{subfigure} \hfill
    \begin{subfigure}[b]{0.53\textwidth}
        \centering
        \begin{algorithmic} 
        \State{$
            \begin{bmatrix}
                x~ & y
            \end{bmatrix}
            \gets
            \begin{bmatrix}
                x_0~ & y_0
            \end{bmatrix}
            $
        }
        \While{$\star$}
            \vspace*{1mm}
            \State{$
                \begin{bmatrix}
                    x \\ y
                \end{bmatrix}
                \gets
                \begin{bmatrix}
                    x + y \cdot \Delta_t \\
                    y + \left(y \cdot \left( 1 - x^2 \right) - x \right) \cdot \Delta_t 
                \end{bmatrix}
                $
            }
            \vspace*{1mm}
        \EndWhile
        \end{algorithmic}
        \caption{A \emph{polynomial} loop, modelling the discrete-time Van der Pol oscillator \cite{Dreossi17} for some constant sampling time $\Delta_t$.}
        \label{fig:intro:examples:poly}
    \end{subfigure}
    \caption{Two examples of deterministic programs.}
    \label{fig:intro:examples}
\end{figure}

\paragraph{Problem setting.}
Let us motivate our hardness results using the two loops in Figure~\ref{fig:intro:examples}, showcasing that very small changes in loop arithmetic may significantly increase the difficulty of computing the strongest invariants.
Figure~\ref{fig:intro:examples:affine} depicts an affine loop, that is, a loop where all updates are affine combinations of program variables. On the other hand, Figure~\ref{fig:intro:examples:poly} shows a polynomial loop whose updates are polynomials in program variables.

An affine (polynomial) invariant is a conjunction of affine (polynomial) equalities holding before and after every loop iteration.
The computability of both the strongest affine and polynomial invariant has been studied extensively.
For affine loops, the seminal paper \cite{Karr76} shows that the strongest affine invariant is computable, whereas \cite{Kovacs08} proves computability of the strongest polynomial invariant for single-path affine loops.
Regarding polynomial programs, for example the one in Figure~\ref{fig:intro:examples:poly}, \cite{Muller-OlmS04-2} gives an algorithm to compute all polynomial invariants of \emph{bounded degree}.

Based on these results, the strongest polynomial invariant of Figure~\ref{fig:intro:examples:affine} is thus computable.
Yet, the more general problem of computing the strongest polynomial invariant for \emph{polynomial loops} without any restriction on the degree remained an open challenge since~2004~\cite{Muller-OlmS04}. 
In this paper, we address this challenge, which we coin as the \spinv problem and define below. 

\problembox{The \spinv Problem\label{spinv}: Given a single-path loop with polynomial updates, compute the strongest polynomial invariant.}

In Section~\ref{sec:ptp-spinv}, we prove that \spinv is \emph{very hard}, essentially \lq\lq defending\rq\rq\ the state-of-the-art that so far failed to derive computational bounds on computing the strongest polynomial invariants of polynomial loops.
The crux of our work is based on the \skolem problem, a prominent algebraic problem in the theory of linear recurrences~\cite{Everest-Skolem,Tao-Skolem}, which we briefly recall below and refer to Section~\ref{sec:preliminaries:RecEqs} for details.

\problembox{The \skolem Problem\label{skolem}~\cite{Everest-Skolem,Tao-Skolem}: 
Does a given linear recurrence sequence with constant coefficients have a zero?
}

The decidability of the \skolem problem has been open for almost a century, and its decidability is connected to far-fetching conjectures in number theory \cite{Bilu22,Lipton22}.
In Section~\ref{sec:ptp-spinv}, we show that \spinv is at least as hard as the \skolem problem, providing thus a computational lower bound showcasing the hardness of \spinv. 

To the best of our knowledge, our results from Section~\ref{sec:ptp-spinv} are the first lower bounds for \spinv and provide an answer to the open challenge posed by~\cite{Muller-OlmS04-2}.
While~\cite{Ouaknine23} proved that the strongest polynomial invariant is uncomputable for multi-path polynomial programs, the computability of \spinv has been left open for future work. With our results proving that \spinv is \skolem-hard (Theorem~\ref{theorem:ptp-spinv}), we show that the missing computability proof of \spinv is not surprising: solving \spinv is really hard.

\paragraph{Connecting invariant synthesis and reachability.}
A computational gap also exists in the realm of model-checking between affine and polynomial programs, similar to the computability of \spinv.
Point-to-point reachability is arguably the simplest model-checking property; it asks whether a program can reach a given target state from a given initial state.
For example, one may start the Van der Pol oscillator from Figure~\ref{fig:intro:examples:poly} in some initial configuration $(x_0, y_0)$ and certify that it will eventually reach a certain target configuration $(x_t, y_t)$.
Reachability, and even more involved model-checking properties, are known to be decidable for affine loops \cite{Karimov22}.
However, the decidability or mere reachability of \emph{polynomial loops} remains unknown without any existing non-trivial lower bounds.
We refer to this reachability quest via the \ptp problem.

\problembox{The \hyperref[ptp]{Point-To-Point Reachability}~Problem (\ptp) \label{ptp}: 
Given a single-path loop with polynomial updates, is a given target state reachable starting from a given initial state?
}

In Section~\ref{sec:skolem-ptp}, we resolve the lack of computational results on reachability in polynomial loops.
In particular, we show that \ptp is \skolem-hard (Theorem~\ref{theorem:skolem}) as well.
To reduce \skolem to \ptp, we construct a polynomial loop from a given linear recurrence sequence, such that the loop reaches the all-zero state if and only if the linear recurrence sequence has a zero.
For our reduction, a linear recurrence sequence of order $k$ is encoded as a loop with $k$ variables.
The crux of the reduction in Section~\ref{sec:skolem-ptp} is that every variable is a shifted \lq\lq non-linear variant\rq\rq of the original sequence such that, once any variable becomes $0$, it remains $0$ forever.
Then, the resulting loop reaches the all-zero state if and only if the original sequence has a zero.
To the best of our knowledge, this yields the first non-trivial hardness result for \ptp.

In Section~\ref{sec:ptp-spinv}, we further show that \ptp and \spinv are connected in the sense that \ptp reduces to \spinv.
To reduce \ptp to \spinv, we show how to decide whether a given loop reaches a given target state only using polynomial invariants.
For the reduction, we add an auxiliary variable to the loop that becomes and remains $0$ as soon as the original loop reaches the given target state.
Intuitively, the auxiliary variable is \emph{eventually} invariant if and only if the original loop reaches the target state.
Utilizing techniques from computational algebraic geometry, we show how to decide whether the auxiliary variable is eventually invariant given the strongest polynomial invariant.
Hence, we show that \spinv is at least as hard as \ptp.

Therefore, our reduction chain $\skolem \leq \ptp \leq \spinv$ implies that the decidability of \ptp and/or \spinv would immediately solve the \skolem problem a longstanding conjecture in number theory.

\paragraph{Beyond (non)deterministic loops and invariants.}
In addition to computational limits within standard, (non)deterministic programs, we further establish computational (hardness) bounds in probabilistic loops.
Probabilistic programs model stochastic processes and encode uncertainty information in standard control flow, used for example in cryptography \cite{Barthe2012}, privacy \cite{Barthe2012a}, cyber-physical systems \cite{KofnovMSBB22}, and machine learning \cite{Ghahramani2015}.

Because classical invariants, as in \spinv, do not account for probabilistic information, we provide a proper generalization of the strongest polynomial invariant for probabilistic loops in Section~\ref{sec:prob-invariants} (Lemma~\ref{lemma:moment-inv-generalization}).
With this generalization, we transfer the \spinv problem to the probabilistic setting.
We hence consider the probabilistic version of \spinv as being the \probspinv problem.

\problembox{The \probspinv Problem\label{probspinv}: Given a probabilistic loop with polynomial updates, compute the \lq\lq probabilistic analog\rq\rq\ of the strongest polynomial invariant.}

In Section~\ref{sec:prob-invariants} we prove that \probspinv inherits \skolem-hardness from its classical \spinv analog (Theorem~\ref{thm:spinv-probspinv}). We also show that enriching the probabilistic program model with guards or branching statements renders the strongest polynomial (probabilistic) invariant uncomputable, even in the affine case (Theorems~\ref{thm:prob-invariants-pcp}).
We nevertheless provide a decision procedure when considering \probspinv for a restricted class of polynomial loops:
we define the class of \emph{moment-computable} (polynomial) loops and show that \probspinv is computable for such loops (Algorithm~\ref{alg:moment-inv-ideal}).
Despite being restrictive, our moment-computable loops subsume affine loops with constant probabilistic choice.
As such, Section~\ref{sec:prob-invariants} shows the limits of computability in deriving the strongest polynomial (probabilistic) invariants for probabilistic polynomial loops.

\paragraph{Our contributions.} In conclusion, the main contributions of our work are as follows: 
\begin{itemize}
    \item In Section~\ref{sec:skolem-ptp}, we provide a reduction from \skolem to point-to-point reachability for polynomial loops, proving that \ptp is \skolem-hard (Theorem~\ref{theorem:skolem}).

    \item Section~\ref{sec:ptp-spinv} gives a reduction from \ptp to the problem of computing the strongest polynomial invariant of polynomial loops, establishing the connection between \ptp and \spinv. As such, we prove that \spinv is \skolem-hard (Theorem~\ref{theorem:ptp-spinv}).

    \item In Section~\ref{sec:prob-invariants}, we generalize the concept of strongest polynomial invariants to the probabilistic setting (Lemma~\ref{lemma:moment-inv-generalization}). We show that \probspinv is \skolem-hard (Theorem~\ref{thm:spinv-probspinv}) and uncomputable for general polynomial probabilistic programs (Theorem~\ref{thm:prob-invariants-pcp}), but it becomes computable for moment-computable polynomial probabilistic programs (Algorithm~\ref{alg:moment-inv-ideal}).
\end{itemize}

\section{Preliminaries \label{section:preliminaries}}

Throughout the paper, we write \N for the natural numbers, \Q for the rationals, \R for the reals, and \Qbar for the algebraic numbers.
We denote by $\K[x_1, \ldots, x_k]$ the polynomial ring over $k$ variables with coefficients in some field \K.
Further, we use the symbol \P\ for probability measures and \E\ for the expected value operator.

\subsection{Program Models \label{sec:preliminaries:program-models}}

In accordance with \cite{Ouaknine23,Varonka23}, we consider \emph{polynomial programs} $\prog = (Q, E, q_0)$ over $k$ variables, where $Q$ is a set of locations, $q_0 \in Q$ is an initial location, and $E \subseteq Q \times \Q[x_1, \ldots, x_k]^k \times Q$ is a set of transitions. 
The vector of \emph{variable valuations} is denoted as $\vec{x} = (x_1, \ldots, x_k)$, where each transition $(q, f, q') \in E$ maps a (program) configuration $(q, \vec{x})$ to some configuration $(q', f(\vec{x}))$.
A transition $(q, f, q') \in E$ is \emph{affine} if the function $f$ is affine. 
In case all program transitions $(q, f, q') \in E$ are affine, we say that the polynomial program $\prog$ is an \emph{affine program}.

A \emph{loop} is a program $\mathcal{L} = (Q, E, q_0)$ with exactly two locations $Q = \{ q_0, q_1 \}$, such that the initial state $q_0$ has exactly one outgoing transition to $q_1$ and all outgoing transitions of $q_1$ are self-loops, that is, $E = \{ (q_0, f_1, q_1), (q_1, f_2, q_1), \ldots, (q_1, f_n, q_1) \}$.

In a \emph{guarded program}, each transition is additionally guarded by an equality/inequality predicate among variables of the state vector $\vec{x}$.
If in some configuration the guard of an outgoing transition holds, we say that the transition is \emph{enabled}, otherwise the transition is \emph{disabled}.

\paragraph{(Non)Deterministic programs.}
If for any location $q \in Q$ in a program $\prog$ there is exactly one outgoing transition $(q, f, q')$, then $\prog$ is \emph{deterministic}; otherwise $\prog$ is \emph{nondeterministic}.
A deterministic guarded program may have multiple outgoing transitions from each location, but for any configuration, exactly one outgoing transition must be enabled.
For a guarded nondeterministic program, we require that each configuration has at least one enabled outgoing transition. 
Deterministic, unguarded programs are called \emph{single-path} programs.

To capture the concept of a loop invariant, we consider the collecting semantics of $\prog$, associating each location $q \in Q$ with a set of vectors $\Scal_q$ that are reachable from the initial state $(q_0, \vec{0})$.
More formally, the sets $\{ \Scal_q \mid q \in Q \}$ are the least solution of the inclusion system
\[
    \Scal_{q_0} \supseteq \{ \vec{0} \} \qquad \text{and} \qquad \Scal_{q'} \supseteq f(\Scal_q) \quad \text{for all } (q, f, q') \in E.
\]
\begin{definition}[Invariant]\label{def:inv}
A polynomial $p \in \overline{\Q}[x_1, \ldots, x_k]$ is an \emph{invariant} with respect to program location $q \in Q$, if for all reachable configurations $\vec{x} \in \Scal_q$ the polynomial vanishes, that is $p(\vec{x}) = 0$.
Moreover, for a loop $\mathcal{L}$, the polynomial $p$ is an \emph{invariant of $\mathcal{L}$}, if $p$ is an invariant with respect to the looping state $q_1$.
\end{definition}

\paragraph{Probabilistic programs.}
In probabilistic programs, a probability $pr$ is added to each program transition. That is, $E \subseteq Q \times \Q[x_1, \ldots, x_k]^k \times (0,1] \times Q$, where we require that each location has countably many outgoing transitions and that their probabilities $pr$ sum up to $1$.
Under the intended semantics, a transition $(q, f, pr, q')$ then maps a configuration $(q, \vec{x})$ to configuration $(q', f(\vec{x}))$ with probability $pr$. 
Again, for guarded probabilistic programs, we require that each configuration has at least one enabled outgoing transition and that the probabilities of the enabled transition sum up to $1$.

For probabilistic programs $\prog$, we consider moment invariants over higher-order statistical moments of the probability distributions induced by $\prog$ (see Section~\ref{sec:prob-invariants}).
In this respect, it is necessary to count the number of executed transitions in the semantics of $\prog$.
Formally, the sets $\{ \Scal_q^n \mid q \in Q, n \in \N_0 \}$ are defined as 
\[
    \Scal_{q_0}^0 \coloneqq \{ \vec{0} \} \qquad \text{and} \qquad \Scal_{q'}^{n+1} \coloneqq f \left( \Scal_{q}^n \right) \quad \text{for all } (q, f, pr, q') \in E \text{ and } n \in \N_0.
\]
In addition, the probability of a configuration $\vec{x}$ in location $q$ after $n$ iterations, in symbols $\P(\vec{x} \mid \Scal_q^n)$, can be defined inductively:
(i) in the initial state, the configuration $\vec{0}$ after $0$ executed transitions has probability $1$; 
(ii) for any other state, the probability of reaching a specific configuration is defined by summing up the probabilities of all incoming paths.
More formally, the probability $\P(\vec{x} \mid \Scal_q^n)$ is 
\[
    \P \left( \vec{x} \mid \Scal_q^0 \right) \coloneqq \begin{cases}
        1 & q = q_0 \land \vec{x} = \vec{0} \\
        0 & \text{otherwise}
    \end{cases} \qquad \text{and} \qquad \P \left( \vec{x} \mid \Scal_{q'}^{n+1} \right) \coloneqq \sum_{(q,f,pr,q') \in E} \ \sum_{\vec{y} \in f^{-1}(\vec{x})} pr \cdot \P(\vec{y} \mid \Scal_{q}^{n}).
\]
We then define the $n$th higher-order statistical moment of a monomial $M$ in program variables as the expected value of $M$ after $n$ loop iterations. Namely,
\begin{equation}\label{eq:EMn}
    \E[M_n] \coloneqq \sum_{q \in Q} \sum_{\vec{x} \in \Scal_{q}^{n}} M(\vec{x}) \cdot \P(\vec{x} \mid \Scal_{q}^{n}), 
\end{equation}
where $M(\vec{x})$ evaluates the monomial $M$ in a specific configuration $\vec{x}$. 

\begin{example}\label{ex:prob-loop}
The following loop encodes a symmetric 1-dimensional random walk starting at $0$.
In every step, the random walk moves left or right with probability $\nicefrac{1}{2}$.
The loop is given in code:

\begin{minipage}{\linewidth}
\begin{algorithmic}
\State{$x \gets 0$
}
\While{$\star$}
    \vspace*{1mm}
    \State{$x \gets x + 1 \ \left[ \nicefrac{1}{2} \right] \ x-1$}
    \vspace*{1mm}
\EndWhile
\end{algorithmic}
\end{minipage}

Replacing the probabilistic choice in the loop body with non-deterministic choice, results in a non-deterministic program.
\end{example}

\paragraph{Universality of loops.}
In this paper, we focus on polynomial loops.
This is justified by the universality of loops \cite[Section~4]{Ouaknine23}, as every polynomial program can be transformed into a polynomial loop that preserves the collecting semantics.
Intuitively, this is done by merging all program states into the looping state and by introducing additional variables that keep track of which state is actually active while invalidating infeasible traces.
It is then possible to recover the sets $\Scal_q^{(n)}$ of the original program from the sets $\Scal_q^{(n)}$ of the loop.

\subsection{Computational Algebraic Geometry \& Strongest Invariants\label{sec:preliminaries:alg-geom}}

We study polynomial invariants $p(\vec{x})$ of polynomial programs;
here, $p(\vec{x})$ are multivariate polynomials in program variables $\vec{x}$. 
We therefore recap necessary terminology from algebraic geometry~\cite{CoxLittleOshea97}, to support us in reasoning whether $p(\vec{x})=0$ is a loop invariant.
In the following $\K$ denotes a field, such as $\R$, $\Q$ or $\overline{\Q}$.

\begin{definition}[Ideal]\label{def:ideal}
A subset of polynomials $I \subseteq \K[x_1, \ldots, x_k]$ is an \emph{ideal} if (i) $0 \in I$; (ii) for all $x,y \in I$: $x+y \in I$; and (iii) for all $x \in I$ and $y \in \K[x_1, \ldots, x_k]$: $xy \in I$.
For polynomials $p_1, \ldots, p_l \in \K[x_1, \ldots, x_k]$ we denote by $\langle p_1, \ldots, p_l \rangle$ the ideal generated by these polynomials, that is
\begin{equation*}
\langle p_1, \ldots, p_l \rangle := \left\{ \sum_{i = 1}^l q_i p_i \ \middle\vert\ q_1, \ldots q_k \in \K[x_1, \ldots, x_k] \right\}
\end{equation*}

The set $I = \langle p_1, \ldots, p_l \rangle$ is an ideal, with the polynomials $p_1, \ldots, p_l$ being a \emph{basis} of $I$.
\end{definition}

Of particular importance to our work is the set of all polynomial invariants of a program location.
It is easy to check that this set forms an ideal.

\begin{definition}[Invariant Ideal]\label{def:invIdeal}
Let $\mathcal{P}$ be a program with location $q$.
The set $\mathcal{I}$ of all invariants with respect to the location $q$ is called the \emph{invariant ideal} of $q$.
If $\mathcal{P}$ is a loop and $\mathcal{I}$ is the invariant ideal with respect to the looping state $q_1$, we call $\mathcal{I}$ the invariant ideal of the loop $\mathcal{P}$
\footnote{
Computing bases for invariant ideals is equivalent to computing the \emph{Zariski closure} of the loop:
the Zariski closure is the smallest algebraic set containing the set of reachable states~\cite{DBLP:conf/lics/HrushovskiOP018}.}.
\end{definition}

As the invariant ideal $\mathcal{I}$ of a loop $\mathcal{L}$ contains \emph{all} polynomial invariants, a basis for $\mathcal{I}$ is the strongest polynomial invariant of $\mathcal{L}$.
This is further justified by the following key result, establishing that every ideal has a basis.

\begin{theorem}[Hilbert's Basis Theorem] 
Every ideal $I \subseteq \K[x_1, \ldots, x_k]$ has a basis. 
That is, $I = \langle p_1, \ldots, p_l \rangle$ for some $p_1, \ldots, p_l \in I$.
\end{theorem}

While an ideal $I$ may have infinitely many bases, the work of~\cite{Buchberger-thesis} proved that every ideal $I$ has a {unique} (reduced) \emph{Gröbner basis}, where uniqueness is guaranteed modulo some \emph{monomial order}. 
A {monomial order $<$} is a total order on all monomials such that for all monomials $m_1, m_2, m_3$, if $m_1 < m_2$ then $m_1 m_3 < m_2 m_3$.
For instance, assume our polynomial ring is $\K[x, y, z]$, that is, over three variables $x$, $y$, and $z$.
A total order $z < y < x$ over variables can be extended to a lexicographic ordering on monomials, denoted also by $<$ for simplicity.
In this case, for example, $x y z^3 < x y^2$ and $y^2z < x$.
For a given monomial order, one can consider the leading term of a polynomial $p$ which we denote by $LT(p)$.
For a set of polynomials $S$ we write $LT(S)$ for the set of all leading terms of all polynomials.
Continuing the  example mentioned before, we have $LT(xyz^3 + xy^2) = xy^2$ and $LT(\{ y + z, y^2z + x \}) = \{ y, x \}$.

\begin{definition}[Gröbner Basis]
Let $I \subseteq \K[x_1, \ldots, x_k]$ be an ideal and fix a monomial order.
A basis $G = \{ g_1, \ldots, g_k \}$ of $I$ is a \emph{Gröbner basis}, if $\langle LT(g_1), \ldots, LT(g_l) \rangle = \langle LT(I) \rangle$.
Further, $G$ is a \emph{reduced Gröbner basis} if every $g_i$ has leading coefficient $1$ and for all $g, h \in G$ with $g \neq h$, no monomial in $g$ is a multiple of $LT(h)$.
\end{definition}

Gröbner bases provide the workhorses to compute and implement algebraic operations over (infinite) ideals, including ideal intersections/unions, variable eliminations, and polynomial memberships.
Given \emph{any} basis for an ideal $I$, a unique reduced Gröbner basis with respect to any monomial ordering $<$ is computable using Buchberger's algorithm~\cite{Buchberger-thesis}.
A central property of Gröbner basis computation is that repeated division of a polynomial $p$ by elements of a Gröbner basis results in a unique remainder, regardless of the order in which the divisions are performed.
Hence, to decide if a polynomial $p$ is an element of an ideal $I$, that is deciding polynomial membership, it suffices to divide $p$ by a Gröbner basis of $I$ and check if the remainder is $0$. 
Moreover, eliminating a variable $y$ from an ideal $I \subseteq \K[x,y]$ is performed by computing the Gr\"ober basis of the elimination ideal $I \cap \K[x]$ only over $x$.

\subsection{Recurrence Equations}\label{sec:preliminaries:RecEqs}

Recurrence equations relate elements of a sequence to previous elements.
There is a strong connection between recurrence equations and program loops:
assignments in program loops relate values of program variables in the current iteration to the values in the next iteration. It is therefore handy to interpret a (polynomial) program loop as a recurrence.
We briefly introduce linear and polynomial recurrence systems and refer to~\cite{kauers2011concrete} for details.

We say that a sequence $u(n): \N_0 \to \Q$ is a \emph{linear recurrence sequence (\lrs)} of order $k$, if there are coefficients $a_0, \ldots, a_{k-1} \in \Q$, where $a_0 \neq 0$ and for all $n \in \N_0$ we have
\begin{equation}\label{eq:Cfinite}
    u(n+k) = a_{k-1} u(n+k{-}1) + \ldots + a_1 u(n+1) + a_0 u(n)
\end{equation}
The recurrence equation~\eqref{eq:Cfinite} is called a \emph{linear recurrence equation}, with the coefficients $a_0, \ldots, a_{k-1}$ and the initial values $u(0), \ldots, u(k{-}1)$ uniquely specifying the sequence $u(n)$.
Any \lrs $u(n)$ of order $k$ as defined via~\eqref{eq:Cfinite} can be specified by a system of $k$ linear recurrence sequences $u_1(n), \ldots, u_k(n)$, such that each $u_i(n)$ is of order $1$ and, for all $n \in \N_0$, we have $u(n) = u_1(n)$ and
\begin{align}
    u_1(n+1) &= \sum_{i=1}^k a^{(1)}_i u_i(n) = a^{(1)}_1 u_1(n) + \ldots + a^{(1)}_k u_k(n) \nonumber \\
    & \vdots \label{eq:CFiniteRecs}\\
    u_k(n+1) &= \sum_{i=1}^k a^{(k)}_i u_i(n) = a^{(k)}_1 u_1(n) + \ldots + a^{(k)}_k u_k(n) \nonumber
\end{align}
Again, the \lrs $u(n)$ is uniquely defined by the coefficients $a_i^{(j)}$ and the initial values $u_1(0), \ldots, u_k(0)$.

\emph{Polynomial recursive sequences} are natural generalizations of linear recurrence sequences and allow not only linear combinations of sequence elements but also polynomial combinations \cite{CadilhacMPPS20}.
More formally, a sequence $u(n)$ is \emph{polynomial recursive}, if there exists $k \in \N$ sequences $u^1(n), \ldots, u^k(n): \N_0 \to \Q$ such that $u(n) = u_1(n)$ and there are polynomials $p_1, \ldots, p_k \in \Q[u_1, \ldots, u_k]$ such that, for all $n \in \N_0$, we have
\begin{align}
    u_1(n+1) &= p_1( u_1(n), \ldots, u_k(n)) \nonumber \\
    &\vdots \label{eq:PolyRec} \\
    u_k(n+1) &= p_k(u_1(n), \ldots, u_k(n)) \nonumber
\end{align}
The sequence $u(n)$ from~\eqref{eq:PolyRec} is uniquely defined by the polynomials $p_1, \ldots, p_k$ and the initial values $u_1(0), \ldots, u_k(0)$.
In contrast to linear recurrence sequences~\eqref{eq:Cfinite}, polynomial recursive sequences~\eqref{eq:PolyRec} \emph{cannot} be in general modeled using a single polynomial recurrence~\cite{CadilhacMPPS20}.
Systems of recurrences are widely used to model the evolution of dynamical systems in discrete time.

We conclude this section by recalling the \skolem problem~\cite{Bilu22,Lipton22} related to linear recurrence sequences, whose decidability is an open question since the 1930s.
We formally revise the definition from Section~\ref{sec:intro} as:

\problembox{The \skolem Problem \cite{Everest-Skolem,Tao-Skolem}: Given an \lrs $u(n), n \in \N_0$, does there exist some $m \in \N_0$ such that $u(m) = 0$?}

In the upcoming sections, we show that the \skolem problem is reducible to the decidability of three fundamental problems in programming languages, namely \ptp{}, \spinv{} and \probspinv{} from Section~\ref{sec:intro}.
As such, we prove that the \skolem problem gives us intrinsically hard computational lower bounds for \ptp{}, \spinv{}, and \probspinv{}.

\section{Hardness of Reachability in Polynomial Programs}\label{sec:skolem-ptp}

We first address the computational limitations of reachability analysis within polynomial programs. 
It is decidable whether a loop with \emph{affine} assignments reaches a target state from a given initial state~\cite{KannanL80}.
Additionally, even problems generalizing reachability are known to be decidable for linear loops, such as various model-checking problems~\cite{Karimov22}.
However, reachability for loops with polynomial assignments, or equivalently discrete-time polynomial dynamical systems, has been an open challenge. In this section, we address this reachability challenge via our \ptp{} problem, showing that reachability in polynomial program loops is at least as hard as the \skolem problem (Theorem~\ref{theorem:skolem}). 
To this end, let us revisit and formally define our \ptp problem from Section~\ref{sec:intro}, as follows.

\problembox{The Point-To-Point Reachability Problem (\ptp): Given a system of $k$ polynomial recursive sequences $u_1(n), \ldots, u_k(n), n \in \N_0$ and a target vector $\vec{t} = (t_1, \ldots, t_k)$, does there exist some $m \in \N_0$ such that for all $1 \leq i \leq k$, it holds that $u_i(m) = t_i$?}
To the best of our knowledge, nothing is known about the hardness of \ptp for polynomial recursive sequences\footnote{For linear systems, the Point-To-Point Reachability problem (\ptp) is also referred to as the \emph{Orbit problem} in~\cite{KannanL80}.}, and hence for loops with arbitrary polynomial assignments, apart from the trivial lower bounds provided by the linear/affine cases~\cite{KannanL80,Karimov22}.

In the sequel, in Theorem~\ref{theorem:skolem} we prove that the \ptp problem for polynomial recursive sequences is \emph{at least as hard} as \skolem.
Doing so, we show that solving \skolem can be solved by \emph{reducing} it to inputs for \ptp, written in symbols as $\skolem \leq \ptp$.
We thus establish a computational lower bound for \ptp in the sense that providing a decision procedure for \ptp for polynomial recursive sequences would prove the decidability of the long-lasting open decision problem given by \skolem.

\paragraph{Our reduction for $\skolem\leq\ptp$.}
In a nutshell, we fix an arbitrary \skolem instance, that is, a linear recurrence sequence $u(n)$ of order $k$.
We say that the instance $u(n)$ is \emph{positive}, if there exists some $m \in \N_0$ such that $u(m) = 0$, otherwise we call the instance \emph{negative}.
Our reduction $\skolem \leq \ptp$ constructs an instance of \ptp that reaches the all-zero vector $\vec{0}$ if and only if the \skolem instance is positive.
Hence, a decision procedure for \ptp would directly lead to a decision procedure for \skolem.

Following~\eqref{eq:Cfinite}, let our \skolem instance of order $k$ to be the \lrs $u(n): \N_0 \to \Q$ specified by coefficients $a_0, \ldots a_{k-1} \in \Q$ such that $a_0 \neq 0$ and, for all $n \in \N_0$, we have
\begin{equation}\label{eq:skolem:instance}
    u(n + k) = a_{k-1} \cdot u(n+k-1) + \ldots + a_1 \cdot u(n+1) + a_0 \cdot u(n) = \sum_{i = 0}^{k-1} a_{i} \cdot u(n+i).
\end{equation}
From our \skolem instance~\eqref{eq:skolem:instance}, we construct a system of $k$ polynomial recursive sequences $x_0, \ldots, x_{k-1}$, as given in~\eqref{eq:PolyRec}. Namely, the initial sequence values are defined inductively as
\[
    \boxed{x_0(0) \coloneqq u(0)} \qquad 
    \boxed{x_i(0) \coloneqq u(i) \cdot \prod_{\ell = 0}^{i-1} x_\ell(0) \qquad (1 \leq i < k)}
\]
With the initial values defined, the sequences $x_0, \ldots, x_{k-1}$ are uniquely defined via the following system of recurrence equations:
\begin{equation}\label{strongest-invariant:hardness:simulation}
        \boxed{
            x_i(n+1) \coloneqq x_{i+1}(n) \qquad (1 \leq i < k-1)
        } \qquad \boxed{
        x_{k-1}(n+1) \coloneqq \sum_{i = 0}^{k-1} a_{i} \cdot x_i(n) \cdot \prod_{\ell=i}^{k-1} x_\ell(n)
        }
\end{equation}
Intuitively, the $x_i$ sequences are \lq\lq non-linear variants\rq\rq\ of the \skolem instance $u(n)$ such that, once any $x_i$ reaches $0$, $x_i$ remains $0$ forever.
The target vector for our \ptp instance is therefore $\vec{t} = \vec{0}$.

Let us illustrate the main idea of our construction with the following example.

\begin{example} \label{ex:skolem-reduction}
Assume our \skolem instance from~\eqref{eq:skolem:instance} is given by the recurrence $u(n{+}3) = 2u(n{+}2) -2u(n{+}1) -12u(n)$ and the initial values $u(0) = 2, u(1) = -3, u(2) = 3$.
Following our reduction \eqref{strongest-invariant:hardness:simulation}, we construct a system of polynomial recursive sequences $x_i(n)$:
\begin{align*}
    x_0(0) &= u(0) = 2 \quad & x_0(n+1) &= x_1(n) \\
    x_1(0) &= u(1) x_0(0) = -6 \quad & x_1(n+1) &= x_2(n) \\
    x_2(0) &= u(2) x_0(0) x_1(0) = -36 \quad & x_2(n+1) &= 2 x_{2}(n)^2 -2 x_{1}(n)^2 x_2(n) -12 x_0(n)^2 x_1(n) x_2(n)
\end{align*}
The first few sequence elements of $u(n)$ and $x_0(n)$ are shown in Figure~\ref{fig:skolem-ptp-example} and illustrate the key property of our reduction:
\begin{itemize}
\item[(i)] $x_0(n)$ is non-zero as long as $u(n)$ is non-zero, which we prove in Lemma~\ref{strongest-invariant:hardness:simulation-lemma}; 
\item[(ii)] if there is an $N$ such that $u(N) = 0$, it holds that for all $n \geq N: x_0(n) = 0$.
The other sequences $x_1$ and $x_2$ in the system are ``shifted'' variants of $x_0$.
Hence, the constructed sequences all eventually reach the all-zero configuration and remain there.
In Theorem~\ref{theorem:skolem}, we prove that this is the case if and only if the  \skolem instance $u(n)$ is positive.
\end{itemize}
\begin{figure}[tb]
    \centering
    \includegraphics[width=\textwidth]{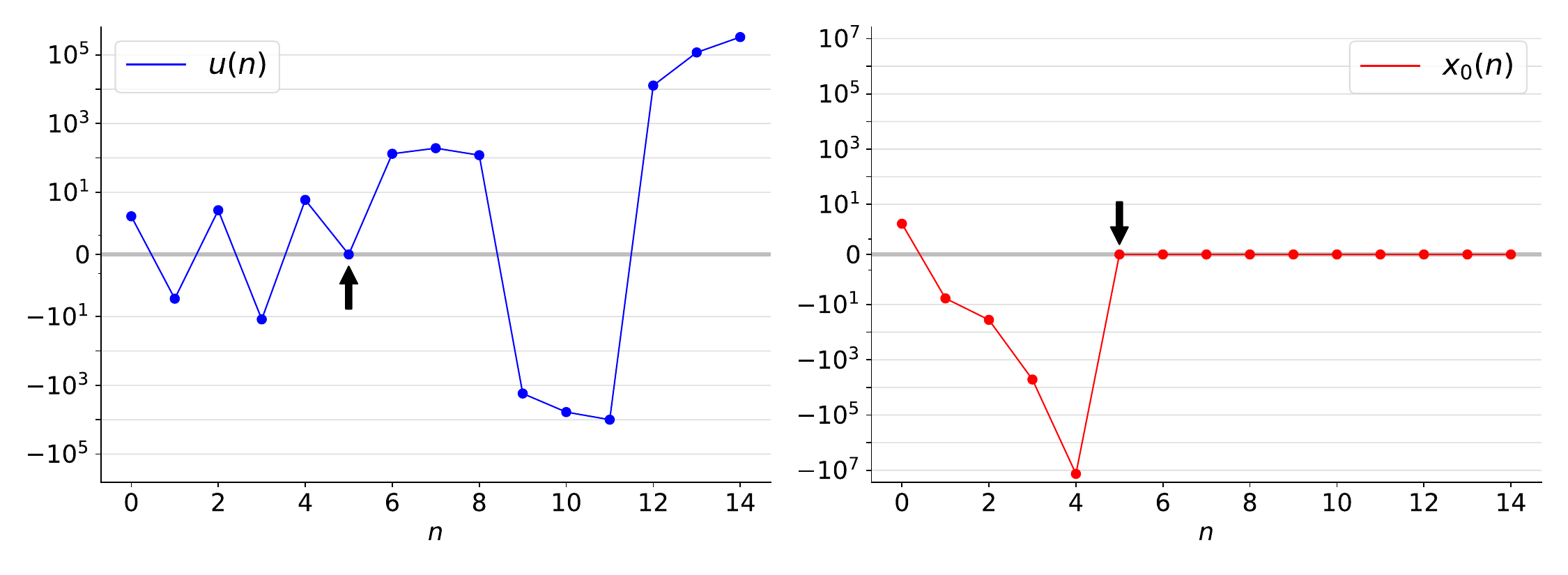}
    \caption{The first $15$ sequence elements of $u(n)$ and $x_0(n)$ in Example~\ref{ex:skolem-reduction}.}
    \label{fig:skolem-ptp-example}
\end{figure}
\end{example}

\paragraph{Correctness of $\skolem\leq\ptp$.} 
To prove the correctness of our reduction $\skolem \leq \ptp$ and to assert the properties (i)-(ii) of Example~\ref{ex:skolem-reduction} among $u(n)$ and $x_i(n)$, we introduce $k$ auxiliary variables $s_0, \ldots, s_{k-1}$ defined as
\begin{equation*}
    \boxed{\begin{aligned}
        s_i(0) \coloneqq \begin{cases}
            1 & (i = 0) \\
            \prod_{\ell = 0}^{i-1} x_\ell(0) & (1 \leq i < k)
        \end{cases}
    \end{aligned}} \qquad
    \boxed{\begin{aligned}
        s_i(n+1) \coloneqq \begin{cases}
            s_{i+1}(n) & (i \neq k-1) \\
            s_{k-1}(n) \cdot x_{k-1}(n) & (i = k-1)
        \end{cases}
    \end{aligned}}
\end{equation*}

As illustrated in Example~\ref{ex:skolem-reduction}, the high-level idea of our reduction is that the $x_i$ sequences are \lq\lq non-linear variants\rq\rq\ of the \skolem instance $u(n)$ such that, once any $x_i$ reaches $0$, $x_i$ remains $0$ forever.
With the next lemma, we make the connections between the sequences $x_i(n)$ and $u(n)$ precise, using the auxiliary sequences $s_i(n)$.
The central connection is $x_0(n) = s_0(n) \cdot u(n)$ and $s_0(n) = \prod_{l=0}^{n-1} x_0(l)$, which we utilize in the correctness proof in Theorem~\ref{theorem:skolem}.
The main idea behind the construction of the \ptp-instance is to ensure that this connection, and similar connections for the other sequences $x_i$ and $s_i$, do hold.
We formally prove these properties by induction.

\begin{lemma}\label{strongest-invariant:hardness:simulation-lemma}
For the system of polynomial recursive sequences in \eqref{strongest-invariant:hardness:simulation}, it holds that $\forall n \geq 0$ and $0 \leq i < k$
\begin{align}
    x_i(n) &= s_i(n) \cdot u(n + i), \text{ and } \label{strongest-invariant:hardness:simulation:eq1} \\
    s_i(n) &= \prod_{\ell = 0}^{n-1}  x_0(\ell)  \cdot \prod_{\ell = 0}^{i-1}  x_\ell(n). \label{strongest-invariant:hardness:simulation:eq2}
\end{align}
\end{lemma}
\begin{proof}
We prove the two properties by well-founded induction on the lexicographic order $(n, i)$, where $n \geq 0$ and $0 \leq i < k$. 
Here, $(n, i) \leq (n', i')$ if and only if $n < n'$ or $n = n' \ \land \ i < i'$.
The order has the unique least element $(0,0)$. 

\noindent\emph{Base case: $n = 0$.}
If $i = 0$, then properties \eqref{strongest-invariant:hardness:simulation:eq1} and \eqref{strongest-invariant:hardness:simulation:eq2} hold by definition of $s_0(0) \coloneqq 1 = \prod_{\ell = 0}^{-1}  x_0(\ell) \cdot \prod_{\ell = 0}^{-1}  x_\ell(0)$ and $x_0(0) \coloneqq u(0) = s_0(0) \cdot u(0)$.
Also, if $0 < i < k$, then properties \eqref{strongest-invariant:hardness:simulation:eq1} and \eqref{strongest-invariant:hardness:simulation:eq2} are trivially satisfied by the definition of the initial values: $s_i(0) \coloneqq \prod_{\ell = 0}^{i-1}  x_\ell(0)$ and $x_i(0) \coloneqq u(i) \cdot \prod_{\ell = 0}^{i-1} x_\ell(0) = u(i) \cdot s_i(0)$.

\noindent\emph{Induction step -- Case 1: $n > 0 \ \land \ 0 \leq i < k{-}1$.}
By the lexicographical ordering, it holds that $(n, i{+}1) < (n{+}1, i)$. Hence, we can assume that properties \eqref{strongest-invariant:hardness:simulation:eq1} and \eqref{strongest-invariant:hardness:simulation:eq2} hold for $(n, i {+} 1)$.
Thus, we have the induction hypothesis
\begin{align}
        x_{i+1}(n) &= s_{i+1}(n) \cdot u(n + i + 1), \text{ and } \label{strongest-invariant:hardness:simulation:step2:eq1} \\
        s_{i+1}(n) &= \prod_{\ell = 0}^{n-1} x_0(\ell) \cdot \prod_{\ell = 0}^{i}  x_\ell(n). \label{strongest-invariant:hardness:simulation:step2:eq2}
\end{align}
To prove property \eqref{strongest-invariant:hardness:simulation:eq1} for $(n{+}1, i)$ means to show that 
\[
    x_{i}(n+1) = s_{i}(n+1) \cdot u(n + i + 1).
\]
The sequences $x_i$ and $s_i$ are defined by $x_i(n{+}1) = x_{i{+}1}(n)$ and $s_i(n{+}1) = s_{i{+}1}(n)$ and hence property \eqref{strongest-invariant:hardness:simulation:eq1} follows from the induction hypothesis \eqref{strongest-invariant:hardness:simulation:step2:eq1}.

To prove property \eqref{strongest-invariant:hardness:simulation:eq2} for $(n{+}1,i)$ means to show that
\[
    s_i(n + 1) = \prod_{\ell = 0}^{n} x_0(\ell) \cdot \prod_{\ell = 0}^{i-1}  x_\ell(n + 1).
\]
We prove the equation by using the induction hypothesis \eqref{strongest-invariant:hardness:simulation:step2:eq2}, the definitions $x_i(n{+}1) = x_{i{+}1}(n)$ and $s_i(n{+}1) = s_{i{+}1}(n)$, and index manipulation:
\begin{align*}
    s_i(n+1) = s_{i+1}(n) &= \prod_{\ell = 0}^{n-1}  x_0(\ell)  \cdot \prod_{\ell = 0}^{i}  x_\ell(n) \\
    &= \prod_{\ell = 0}^{n-1} x_0(\ell) \cdot x_0(n) \cdot \prod_{\ell = 0}^{i-1} x_{\ell + 1}(n) \\
    &= \prod_{\ell = 0}^{n} x_0(\ell) \cdot \prod_{\ell = 0}^{i-1} x_{\ell}(n + 1)
\end{align*}
\noindent\emph{Induction step -- Case 2: $n > 0$ and $i = k{-}1$.}
We show that property \eqref{strongest-invariant:hardness:simulation:eq1} holds for $(n{+}1, k{-}1)$ by proving it to be equivalent to the definition of $x_{k-1}(n{+}1)$.
To do so, we first instantiate property \eqref{strongest-invariant:hardness:simulation:eq1} and replace both $s_{k-1}(n{+}1)$ and $u(n{+}k)$ by their defining recurrence:
\begin{align*}
    x_{k-1}(n+1) 
    &= s_{k-1}(n+1) \cdot u(n+k) \\
    &= s_{k-1}(n) \cdot x_{k-1}(n) \cdot \left( \sum_{i = 0}^{k-1} a_{i} \cdot u(n+i) \right)
\end{align*}
Next, we rearrange and apply the induction hypothesis \eqref{strongest-invariant:hardness:simulation:eq2} for $(n, k{-}1)$ and $(n, i)$ and obtain:
\begin{align*}
    x_{k-1}(n+1) 
    &= x_{k-1}(n) \cdot \left( \sum_{i = 0}^{k-1} a_{i} \cdot u(n+i) \cdot s_{k-1}(n) \right) \\
    &= x_{k-1}(n) \cdot \left( \sum_{i = 0}^{k-1} a_{i} \cdot u(n+i) \cdot \underbrace{\prod_{\ell = 0}^{n-1} x_0(\ell) \cdot \prod_{\ell = 0}^{k-2} x_\ell(n) }_{s_{k-1}(n) \text{ by I.H. \eqref{strongest-invariant:hardness:simulation:eq2}}} \right) \\
    &= x_{k-1}(n) \cdot \left( \sum_{i = 0}^{k-1} a_{i} \cdot u(n+i) \cdot \underbrace{ \prod_{\ell = 0}^{n-1}  x_0(\ell) \cdot \prod_{\ell = 0}^{i-1}  x_\ell(n)}_{=s_i(n) \text{ by I.H. \eqref{strongest-invariant:hardness:simulation:eq2}}} \cdot \prod_{\ell = i}^{k-2} x_\ell(n)  \right) \\
    &= x_{k-1}(n) \cdot \left( \sum_{i = 0}^{k-1} a_{i} \cdot u(n+i) \cdot s_i(n) \cdot \prod_{\ell = i}^{k-2} x_\ell(n) \right) \\
    &= \sum_{i = 0}^{k-1} a_{i} \cdot u(n+i) \cdot s_i(n) \cdot \prod_{\ell = i}^{k-1} x_\ell(n) 
\end{align*}
Now, we can apply the induction hypothesis \eqref{strongest-invariant:hardness:simulation:eq1} to replace $u(n{+}i) \cdot s_i(n)$ by $x_i(n)$ and arrive at the relation:
\[
    x_{k-1}(n+1) = \sum_{i = 0}^{k-1} a_{i} \cdot x_i(n) \cdot \prod_{\ell = i}^{k-1} x_\ell(n)
\]
However, this is exactly the defining recurrence equation from \eqref{strongest-invariant:hardness:simulation}.
Hence, property \eqref{strongest-invariant:hardness:simulation:eq2} necessarily holds for $(n, k{-}1)$.

To prove property \eqref{strongest-invariant:hardness:simulation:eq2} for $(n{+}1, k{-}1)$ we use the defining equation of $s_{k-1}(n{+}1)$ and the induction hypothesis for $(n, k{-}1)$:
\begin{align*}
    s_{k-1}(n+1) 
    &= s_{k-1}(n) \cdot x_{k-1}(n) 
    = x_{k-1}(n) \cdot \prod_{\ell = 0}^{n-1} x_0(\ell)  \cdot \prod_{\ell = 0}^{k-2} x_\ell(n)
    = \prod_{\ell = 0}^{n-1} x_0(\ell) \cdot \prod_{\ell = 0}^{k-1} x_\ell(n) \\
    &= \prod_{\ell = 0}^{n-1} x_0(\ell) \cdot x_0(n) \cdot \prod_{\ell = 0}^{k-2}  x_{\ell + 1}(n)
    = \prod_{\ell = 0}^{n} x_0(\ell) \cdot \prod_{\ell = 0}^{k-2} x_\ell(n+1)
\end{align*}
As we have covered all possible cases, we conclude the proof.
\end{proof}

Lemma~\ref{strongest-invariant:hardness:simulation-lemma} establishes two central properties of our reduction.
We now use these properties to show that \ptp is at least as hard as \skolem.

\begin{theorem}[Hardness of \ptp] \label{theorem:skolem}    
\ptp is \skolem-hard. That is, $\skolem \leq \ptp$.
\end{theorem}
\begin{proof}
We show that our polynomial recursive system constructed in \eqref{strongest-invariant:hardness:simulation} reaches the all-zero vector from the initial value if and only if the original \skolem instance is positive.

$(\Rightarrow):$ Assume the \skolem instance is positive, then there is some smallest $N \in \N_0$ such that $u(N) = 0$.
Property \eqref{strongest-invariant:hardness:simulation:eq1} of Lemma~\ref{strongest-invariant:hardness:simulation-lemma} implies
\[
    x_0(N) = s_0(N) \cdot u(N) = 0.
\]
Using this equation and property \eqref{strongest-invariant:hardness:simulation:eq2} of Lemma~\ref{strongest-invariant:hardness:simulation-lemma}, we deduce that for all $n > N$, each $s_i(n)$ contains $x_0(N)$ as a factor and hence $s_i(n) = 0$. 
Additionally, as $x_i(n) = s_i(n) \cdot u(n{+}i)$ by property \eqref{strongest-invariant:hardness:simulation:eq1}, we conclude that for all $n > N$ also $x_i(n) = 0$.
Hence, the polynomial recursive system reaches the all-zero vector.

$(\Leftarrow)$ Assume that the \skolem instance is negative, meaning that the linear recurrence sequence $u(n)$ does not have a $0$.
In particular, $u(i) \neq 0$ for all $0 \leq i < k$.
Therefore, by definition of the polynomial recursive system \eqref{strongest-invariant:hardness:simulation}, $x_i(0) \neq 0$ for all $0 \leq i < k$.
Towards a contradiction, assume that the polynomial recursive system still reaches the all-zero vector.
Hence, there is a smallest $N \in \N_0$ such that $x_i(N) = 0$ for all $0 \leq i < k$.
In particular, $x_0(N) = 0$.
Moreover, $x_0$ is the last sequence to reach $0$, because of the recurrence equation $x_i(n{+}1) = x_{i+1}(n)$ for $0 \leq i < k$.
Therefore, $N$ is also the smallest number such that $x_0(N) = 0$.
By property \eqref{strongest-invariant:hardness:simulation:eq1} of Lemma~\ref{strongest-invariant:hardness:simulation-lemma}, we have 
\[
    x_0(N) = s_0(N) \cdot u(N) = 0.
\]
However, $s_0(N)$ must be non-zero, because
\[
    s_0(N) = \prod_{\ell = 0}^{N-1} x_0(\ell),
\]
by property \eqref{strongest-invariant:hardness:simulation:eq2} of Lemma~\ref{strongest-invariant:hardness:simulation-lemma}, and the fact that $N$ is the smallest number such that $x_0(N) = 0$. Then we necessarily have $u(N) = 0$, yielding a contradiction.
\end{proof} 

Theorem~\ref{theorem:skolem} shows that \ptp for polynomial recursive sequences is at least as hard as the \skolem problem. 
Thus, reachability and model-checking of loops with polynomial assignments is \skolem-hard. 
A decision procedure establishing decidability for \ptp would lead to a major breakthrough in number theory, as by Theorem~\ref{theorem:skolem} this would imply the decidability of the \skolem problem.

\begin{remark}
In \cite{Ouaknine23} the authors show that the the strongest polynomial invariant is uncomputable for polynomial programs \emph{with nondeterminism}.
The proof reduces from an undecidable problem to finding the strongest polynomial invariant for nondeterministic polynomial programs.
A similarity between our reduction from this section and the reduction in \cite{Ouaknine23} is the idea of projecting specific states to the zero vector.
Nevertheless, the setting and reasons for using such a projection differ significantly between the two reductions.
The reduction in \cite{Ouaknine23} maps invalid program traces to the zero state to argue about the dimension of an algebraic set. 
In contrast, our work maps the single program trace to the zero state if and only if the original \skolem instance is positive.
\end{remark}

\section{Hardness of Computing the Strongest Polynomial Invariant}\label{sec:ptp-spinv}

This section goes beyond reachability analysis and focuses on inferring the strongest polynomial invariants of polynomial loops. As such, we turn our attention to solving the \spinv problem of Section~\ref{sec:intro}, which is formally defined as given below.

\problembox{The \spinv Problem: Given an unguarded, deterministic loop with polynomial updates, compute a basis of its polynomial invariant ideal.}

We prove that finding the strongest polynomial invariant for deterministic loops with polynomial updates, that is, solving \spinv, is at least as hard as \ptp (Theorem~\ref{theorem:ptp-spinv}). Hence, $\ptp \leq \spinv$.

Then, by the $\skolem \leq \ptp$ hardness result of Theorem~\ref{theorem:skolem}, we conclude the \skolem-hardness of \spinv, that is $\skolem \leq \ptp\leq \spinv$.
To the best of our knowledge, our Theorem~\ref{theorem:skolem} together with Theorem~\ref{theorem:ptp-spinv} provide the first computational lower bound on \spinv, when focusing on loops with arbitrary polynomial updates (see Table~\ref{strongest-invariant:table:overview-invariants}).

\paragraph{Our reduction for $\ptp\leq\spinv$.}
We fix an arbitrary \ptp instance of order $k$, given by a system of polynomial recursive sequences $u_1, \ldots, u_k: \N_0 \to \Q$ and a target vector $\vec{t} = (t_1, \ldots, t_k) \in \Q^k$.
This \ptp instance is positive if and only if there exists an $N \in \N_0$ such that $(u_1(N), \ldots, u_k(N)) = \vec{t}$.
For reducing \ptp to \spinv, we construct the following deterministic loop with polynomial updates over $k{+}2$ variables:

\begin{instance}\label{spinv-instance}
\begin{algorithmic}
\State{$
    \begin{bmatrix}
        f~ & g~ & x_1~ & \ldots ~ & x_k
    \end{bmatrix}
    \gets
    \begin{bmatrix}
        1~ & 0~ & u_1(0)~ & \ldots ~ & u_k(0)
    \end{bmatrix}
    $
}
\While{$\star$}
    \vspace*{1mm}
    \State{$
        \begin{bmatrix}
            x_1 \\ \vdots \\ x_k \\ f \\ g
        \end{bmatrix}
        \gets
        \begin{bmatrix}
            p_1(x_1, \ldots, x_k) \\ \vdots \\ p_k(x_1, \ldots, x_k) \\ f \cdot \left( (x_1 - t_1)^2 + \ldots + (x_k - t_k)^2 \right) \\ g + 1
        \end{bmatrix}
        $
    }
    \vspace*{1mm}
\EndWhile
\end{algorithmic}
\end{instance}

The polynomial recursive sequences $u_1, \ldots, u_k$ are fully determined by their initial values and the polynomials $p_1, \ldots, p_k \in \Q[u_1, \ldots, u_k]$ defining the respective recurrence equations $u_i(n{+}1) = p_i(u_1(n), \ldots, u_k(n))$.
Hence, by the construction of the \spinv instance \eqref{spinv-instance}, every program variable $x_i$ models the sequence $u_i$.
As such, for any number of loop iterations $n \in \N_0$, we have $x_i(n) = u_i(n)$.
Moreover, the variable $g$ models the loop counter $n$, meaning $g(n) = n$ for all $n \in \N_0$.
The motivation behind using the program variable $f$ is that $f$ becomes $0$ as soon as all sequences $u_i$ reach their target $t_i$; moreover, $f$ remains $0$ afterward.
More precisely, for $n \in \N_0$, $f(n) = 0$ if and only if there is some $N \leq n$ such that $x_1(N) = t_1 \land \ldots \land x_k(N) = t_k$.
Hence, the sequence $f$ has a $0$ value, and subsequently, all its values are $0$, if and only if the original instance of \ptp is positive.

Let us illustrate the main idea of our $\ptp\leq \spinv$ reduction via the following example.
\begin{example}
Consider the recursive sequences $x(n{+}1) = x(n) + 2$ and $y(n{+}1) = y(n) + 3$, with initial values $x(0) = y(0) = 0$.
It is easy to see that the system $S = (x(n), y(n))$ reaches the target $\vec{t_1} = (4, 6)$ but does not reach the target $\vec{t_2} = (5, 7)$.
Following are the two \spinv instances produced by our reduction for the \ptp instances $(S, \vec{t_1})$ and $(S, \vec{t_2})$.

\medskip

\noindent
\begin{minipage}{\linewidth}
\begin{minipage}{0.45\linewidth}
\textbf{\spinv instance for $\mathbf{(S, \vec{t_1})}$:}
\begin{algorithmic}
\State{$
    \begin{bmatrix}
        f~ & g~ & x~ & y
    \end{bmatrix}
    \gets
    \begin{bmatrix}
        1~ & 0~ & 0~ & 0
    \end{bmatrix}
    $
}
\While{$\star$}
    \vspace*{1mm}
    \State{$
        \begin{bmatrix}
            x \\ y \\ f \\ g
        \end{bmatrix}
        \gets
        \begin{bmatrix}
            x + 2 \\ y+3 \\ f \cdot \left( (x - \mathbf{4})^2 + (y - \mathbf{6})^2 \right) \\ g + 1
        \end{bmatrix}
        $
    }
    \vspace*{1mm}
\EndWhile
\end{algorithmic}
\problembox{
Invariant ideal:
$\langle x - 2g, y - 3g, g(g-1)f \rangle$
}
\end{minipage}\hfill
\begin{minipage}{0.45\linewidth}
\textbf{\spinv instance for $\mathbf{(S, \vec{t_2})}$:}
\begin{algorithmic}
\State{$
    \begin{bmatrix}
        f~ & g~ & x~ & y
    \end{bmatrix}
    \gets
    \begin{bmatrix}
        1~ & 0~ & 0~ & 0
    \end{bmatrix}
    $
}
\While{$\star$}
    \vspace*{1mm}
    \State{$
        \begin{bmatrix}
            x \\ y \\ f \\ g
        \end{bmatrix}
        \gets
        \begin{bmatrix}
            x + 2 \\ y+3 \\ f \cdot \left( (x - \mathbf{5})^2 + (y - \mathbf{7})^2 \right) \\ g + 1
        \end{bmatrix}
        $
    }
    \vspace*{1mm}
\EndWhile
\end{algorithmic}
\problembox{
Invariant ideal:
$\langle x - 2g, y - 3g \rangle$
}
\end{minipage}
\end{minipage}

\noindent
The invariant ideals for both instances are given in terms of Gröbner bases with respect to the lexicographic order for the variable order $g < f < y < x$.

For the instance with the reachable target $\vec{t_1}$, we have $f(n) = 0$ for $n \geq 2$.
Hence, $g(g-1)f$ is a polynomial invariant and must be in the invariant ideal of this \spinv instance; in fact, $g(g-1)f$ is not only in the invariant ideal but even a basis element for the Gröbner basis with the chosen order.
However, $g(g-1)f$ is not in the ideal of the \spinv instance with the unreachable target $\vec{t_2}$. These two \spinv instances illustrate thus how a basis of the invariant ideal can be used to decide \ptp.

While, for simplicity, our recursive sequences $x(n)$ and $y(n)$ are linear, our approach to reducing \ptp to \spinv also applies to polynomial recursive sequences.
In Theorem~\ref{theorem:ptp-spinv}, we show that a polynomial such as $g(g-1)f$ is an element of the basis of the invariant ideal (with respect to a specific monomial order) if and only if the original \ptp instance is positive.
\end{example}

\paragraph{Correctness of $\ptp\leq \spinv$.}
To show that it is decidable whether $f(n)$ has a $0$ given a basis of the invariant ideal, we employ Gröbner bases and an argument introduced in \cite{kauers-thesis} for recursive sequences defined by rational functions, adjusted to our setting using recursive sequences defined by polynomials.

\begin{theorem}[Hardness of \spinv] \label{theorem:ptp-spinv}
    \spinv is at least as hard as \ptp. That is, \ptp $\leq$ \spinv.
\end{theorem}
\begin{proof}
Assume we are given an oracle for \spinv, computing a basis $B$ of the polynomial invariant ideal $\mathcal{I} = \langle B \rangle$ of our loop~\eqref{spinv-instance}.
We show that given such a basis $B$, it is decidable whether $f(n)$ has a root, which is equivalent to the fixed \ptp instance being positive.

Note that by the construction of the loop~\eqref{spinv-instance}, if $f(N) = 0$ for some $N \in \N_0$, then $\forall n \geq N: f(n) = 0$.
Moreover, such an $N$ exists if and only if the \ptp instance is positive.
This is true if and only if there exists an $N \in \N_0$ such that the sequence
\begin{equation*}
    n \mapsto f(n) \cdot n \cdot (n-1) \cdot (n-2) \cdot \ldots \cdot (n - N + 1)
\end{equation*}
is $0$ for all $n \in \N_0$.
Consequently, the polynomial invariant ideal $\mathcal{I}$ contains a polynomial 
\begin{equation} \label{eq:ptp-spinv:grobner-zero-poly}
    P \coloneqq f \cdot g \cdot (g - 1) \ldots \cdot (g - N + 1) 
\end{equation}
for some $N \in \N_0$ only if the \ptp instance~\eqref{spinv-instance} is positive.
It is left to show that, given a basis $B$ of $\mathcal{I}$, it is decidable whether $\mathcal{I}$ contains a polynomial~\eqref{eq:ptp-spinv:grobner-zero-poly}.
Using Buchberger's algorithm \cite{Buchberger-thesis}, $B$ can be transformed into a Gröbner basis with respect to any monomial order.
We choose a total order among program variables such that $g < f < x_1, \ldots, x_k$.
Without loss of generality, we assume that $B$ is a Gröbner basis with respect to the lexicographic order extending the variable order.

In what follows, we argue that if a polynomial $P$ as in~\eqref{eq:ptp-spinv:grobner-zero-poly} is an element of $\mathcal{I}$, then $P$ must be an element of the basis $B$.
As the leading term of $P$ is $g^N \cdot f$, there must be some polynomial $Q$ in $B$ with a leading term that divides $g^N \cdot f$.
By the choice of the lexicographic order, this polynomial must be of the form $Q = Q_1(g) \cdot f - Q_2(g)$, since if any other term would occur in $Q$, it would necessarily be in the leading term.
As both $P \in \mathcal{I}$ and $Q \in \mathcal{I}$, it holds that 
\[
    P \cdot Q_1 - g \cdot (g - 1) \ldots \cdot (g - N + 1) \cdot Q \in \mathcal{I}.
\]
By expanding $P$ and $Q$, we see that the above  polynomial is equivalent to 
\[
    Q_2 \cdot g \cdot (g - 1) \ldots \cdot (g - N + 1).
\]
As this polynomial is in the ideal $\mathcal{I}$, it follows that for all $n \in \N_0$:
\[
    Q_2(n) \cdot n \cdot (n - 1) \ldots \cdot (n - N + 1) = 0.
\]
However, this implies that $Q_2(n)$ has infinitely many zeros, a property that is unique to the zero polynomial.
Therefore, we conclude that $Q_2 \equiv 0$.
Hence, if the original \ptp instance is positive, there necessarily exists a basis polynomial of the form $Q_1(g) \cdot f$.

We show that this basis polynomial $Q_1(g) \cdot f$ actually has the form \eqref{eq:ptp-spinv:grobner-zero-poly}:
choose the basis polynomial of the form $Q_1(g) \cdot f$ such that $Q_1$ has minimal degree.
Assume $Q_1(g)$ is not of the form $g \cdot (g {-} 1) \ldots \cdot (g {-} N {+} 1)$.
Then, at least one factor $(g {-} m)$ is not a factor of $Q_1$, or equivalently $Q_1(m) \neq 0$.
Then, necessarily $f(m) = 0$ and $g \cdot (g {-} 1) \cdot \ldots \cdot (g {-} m {+} 1) \cdot f$ must be in the ideal $\mathcal{I}$, contradicting the minimality of the degree of $Q_1$.

Therefore, we conclude that the \ptp instance is positive if and only if the Gröbner basis contains a polynomial of the form \eqref{eq:ptp-spinv:grobner-zero-poly}.
As the basis $B$ is finite, this property can be checked by enumeration of the basis elements of $B$.
Hence, given an oracle for \spinv, we can decide if the \ptp instance is positive or negative.
\end{proof}

Theorem~\ref{theorem:ptp-spinv} shows that \spinv is at least as hard as the \ptp problem. Together with Theorem~\ref{theorem:skolem}, we conclude that \spinv is \skolem-hard. 

\paragraph{An improved direct reduction from \skolem to \spinv.}
Theorem~\ref{theorem:ptp-spinv} together with Theorem~\ref{theorem:skolem} yields the chain of reductions $$\skolem \leq \ptp\leq \spinv.$$
Within these reductions, a \skolem instance of order $k$ yields a \ptp instance with $k$ sequences, which in turn reduces to a \spinv instance over $k{+}2$ variables.

We conclude this section by noting that, if the linear recurrence sequence of the \skolem-instance is an \emph{integer sequence}, then a reduction directly from \skolem to \spinv can be established by using only $k{+}1$ variables.
A slight modification of $\skolem \leq \ptp$ reduction of Section~\ref{sec:skolem-ptp} results in a reduction from \skolem instances of order $k$ directly to \spinv instances with $k{+}1$ variables.
Any system of polynomial recursive sequences can be encoded in a loop with polynomial updates.
Hence, the instance produced by the $\skolem \leq \ptp$ reduction can be interpreted as a loop.
It is sufficient to modify the resulting loop in the following way:
\begin{equation*}
    \boxed{
        \begin{aligned}
            x_{k-1} & \gets \sum_{i = 0}^{k-1} a_{i} \cdot x_i \cdot \prod_{\ell=i}^{k-1} x_\ell \\
            s_{k-1} & \gets x_{k-1} \cdot s_{k-1}
        \end{aligned}
    } 
    \quad \to \quad 
    \boxed {
        \begin{aligned}
            x_{k-1} & \gets \sum_{i = 0}^{k-1} a_{i} \cdot x_i \cdot \prod_{\ell=i}^{k-1} \mathbf{2} \cdot x_\ell \\
            s_{k-1} & \gets \mathbf{2} \cdot x_{k-1} \cdot s_{k-1}
        \end{aligned}
    }
\end{equation*}

As in the reduction in Section~\ref{sec:skolem-ptp}, the equation $u_0(n) = \frac{x_0(n)}{s_0(n)}$ still holds and the resulting loop reaches the all-zero configuration if and only if the original \skolem-instance is positive (the integer sequence has a $0$).
Additionally, the resulting loop has infinitely many \emph{different} configurations if and only if the \skolem instance is positive, as the additional factor in the updates forces a strict increase in $\lvert s_{k-1} \rvert$.
Assuming a solution to \spinv for the constructed loop, that is a basis of the polynomial invariant ideal, it is decidable whether the number of reachable program locations (and its algebraic closure) is finite or not \cite{CoxLittleOshea97}.
Therefore, an oracle for \spinv implies the decidability of \skolem for \emph{integer sequences}, while the chain of reductions $\skolem \leq \ptp \leq \spinv$ is also valid for rational sequences.
For more details, we refer to \cite{julians-thesis}.

\paragraph{Summary of computability results in polynomial (non)determinstic loops.}
We conclude this section by overviewing our computability results in Table~\ref{strongest-invariant:table:overview-invariants}, focusing on the strongest polynomial invariants of (non)deterministic loops and in relation to the state-of-the-art.

\begin{table}[htb]
    \setlength{\tabcolsep}{0.5em}
    \centering
    \resizebox{\textwidth}{!}{%
    \begin{tabular}{|l|l|l||cl|cl|} \hline
        \multicolumn{3}{|l||}{Program Model} & \multicolumn{2}{c|}{Strongest Affine Invariant} & \multicolumn{2}{c|}{Strongest Polynomial Invariant} \\ \hline
        \multirow{4}*{Det.} & \multirow{2}*{Unguarded} & Affine & \checkmark & \cite{Karr76} & \checkmark & \cite{Kovacs08} \\ \hhline{~~|*5-} %
        & & Poly. & \checkmark & \cite{Muller-OlmS04-2} & \skolem-hard &  Theorems \ref{theorem:skolem} \& \ref{theorem:ptp-spinv} \\ \cline{2-7}
        & \multirow{2}*{Guarded ($=, <$)} & Affine & \multicolumn{4}{c|}{\multirow{2}{*}{\xmark \hspace{0.2cm} (Halting Problem)}} \\ \cline{3-3}
        & & Poly. & \multicolumn{4}{c|}{} \\ \hline
        \multirow{4}*{Nondet.} & \multirow{2}*{Unguarded} & Affine & \checkmark & \cite{Karr76} & \checkmark & \cite{Ouaknine23} \\ \cline{3-7}
        & & Poly. & \checkmark & \cite{Muller-OlmS04-2} & \xmark & \cite{Ouaknine23} \\ \cline{2-7}
        & \multirow{2}*{Guarded ($=, <$)} & Affine & \multicolumn{4}{c|}{\multirow{2}{*}{\xmark \hspace{0.2cm} \cite{Muller-OlmS04}}} \\ \cline{3-3}
        & & Poly. & \multicolumn{4}{c|}{} \\ \hline
    \end{tabular}}
    \medskip
    \caption{Summary of computability results for strongest invariants of \emph{nonprobabilistic} polynomial loops, including our own results (Theorems~\ref{theorem:skolem} \& \ref{theorem:ptp-spinv}). With '\checkmark' we denote decidable problems, while '\xmark' denotes undecidable problems.}
    \label{strongest-invariant:table:overview-invariants}
\end{table}
\section{Strongest Invariant for Probabilistic Loops}\label{sec:prob-invariants}

In this section, we finally go beyond (non-)deterministic programs and address computational challenges in probabilistic programming, in particular loops. Unlike the programming models of Section~\ref{sec:skolem-ptp}--\ref{sec:ptp-spinv}, probabilistic loops follow different transitions with different probabilities (cf. Example~\ref{ex:prob-loop}).

Recall that the standard definition of an invariant $I$, as given in Definition~\ref{def:inv}, demands that $I$ holds in \emph{every reachable} configuration and location.
As such, when using Definition~\ref{def:inv} to define an invariant $I$ of a probabilistic loop, the information provided by the probabilities of reaching a configuration within the respective loop is omitted in $I$. However, Definition~\ref{def:inv} captures an invariant $I$ of a probabilistic loop when every probabilistic loop transition is replaced by a nondeterministic transition.
 
Nevertheless, for incorporating probability-based information in loop invariants, Definition~\ref{def:inv} needs to be revised to consider expected values and higher (statistical) moments describing the value distributions of probabilistic loop variables~\cite{Kozen83,McIver05}.
For instance, the symmetric 1-dimensional random walk from Example~\ref{ex:prob-loop} does not have any non-trivial polynomial invariants.
However, considering expected values of program variables, $\E[x] = 0$ is an invariant property of Example~\ref{ex:prob-loop}.
Therefore, in Definition~\ref{def:mom:PInv} we introduce \emph{polynomial moment invariants} to reason about value distributions of probabilistic loops.
We do so by utilizing higher moments of the probability distributions induced by the value distributions of loop variables during the execution (Section~\ref{sec:spinv:MomInv}).
The notion of polynomial moment invariants is the main contribution of this section as it allows us to transfer specific (un)computability results for classical invariants to the probabilistic case.
We prove that polynomial moment invariants generalize classical invariants (Lemma~\ref{lemma:moment-inv-generalization}) and show that the strongest moment invariants up to moment order $\ell$ are computable for the class of so-called moment-computable polynomial loops (Section~\ref{sec:spinf:comp}).
In this respect, in Algorithm~\ref{alg:moment-inv-ideal} we give a complete procedure for computing the strongest moment invariants of moment-computable polynomial loops. When considering \emph{arbitrary} polynomial probabilistic loops, we prove that the strongest moment invariants are (i) not computable for guarded probabilistic loops (Section~\ref{sec:spinv:guarded}) and (ii) \skolem-hard to compute for unguarded probabilistic loops (Section~\ref{sec:spinv:unguarded}).

\subsection{Polynomial Moment Invariants}\label{sec:spinv:MomInv}
Higher moments capture expected values of monomials over loop variables, for example, $\E[x^2]$ and $\E[xy]$ respectively yield the second-order moment of $x$ and a second-order mixed moment. Such higher moments are necessary to characterize, and potentially recover, the value distribution of probabilistic loop variables, allowing us to reason about statistical properties, such as variance or skewness, over probabilistic value distributions. 

When reasoning about moments of probabilistic program variables, note that in general neither $\E[x^\ell] = \E[x]^\ell$ nor $\E[xy] = \E[x] \E[y]$ hold, due to potential dependencies among the (random) loop variables $x$ and $y$.
Therefore, describing all polynomial invariants among all higher moments by finitely many polynomials is futile.
A natural restriction and the one we undertake in this paper is to consider polynomials over finitely many moments, which we do as follows. 

\begin{definition}[Moments of Bounded Degree]\label{def:bounded-moments}
Let $\ell$ be a positive integer. Then the \emph{set of program variable moments of order at most $\ell$} is given by 
\[
    \E^{\leq \ell} \coloneqq \left \{ \E \left[ x_1^{\alpha_1} x_2^{\alpha_2} \cdots x_k^{\alpha_k} \right] \mid \alpha_1 + \ldots + \alpha_k \leq \ell \right \}.
\]
\end{definition}

Classical invariants are defined over a finite set of program variables.
In the probabilistic setting, the elements of $\E^{\leq \ell}$ serve as the formal variables over which moment invariants are defined.
As such, bounding the degrees of the moments is different from bounding the degrees of the invariants, which is a common technique for classical programs \cite{Muller-OlmS04}.
Although the moments in $\E^{\leq \ell}$ are bounded, in this section, we study unbounded polynomial invariants involving these moments.
While Definition~\ref{def:bounded-moments} uses a bound $\ell$ to define the set of moments of bounded degree, our subsequent results apply to any \emph{finite} set of moments of program variables.

Recall that Section~\ref{sec:preliminaries:program-models} defines the semantics $\Scal_q^n$ of a probabilistic loop with respect to the location $q \in Q$ and the number of executed transitions $n \geq 0$.
The set $\Scal_q^n$ in combination with the probability of each configuration allows us to define the moments of program variables after $n$ transitions.
Further, for a monomial $M$ in program variables, we defined $\E[M_n]$ in~\eqref{eq:EMn} to be the expected value of $M$ after $n$ transitions.
For example, $\E[x_n]$ denotes the expected value of the program variable $x$ after $n$ transitions.
With this, we define the set of polynomial invariants among moments of program variables, as follows.

\begin{definition}[Moment Invariant Ideal]\label{def:moment-ideal}\label{def:mom:PInv} 
Let $\E^{\leq \ell} = \{ \E[M^{(1)}], \ldots, \E[M^{(k)}] \}$ be the set of program variable moments of order less than or equal to $\ell$.
The \emph{moment invariant ideal} $\mathbb{I}^{\leq \ell}$ is defined as
\[
    \mathbb{I}^{\leq \ell} = \left \{ p \left(\E[M^{(1)}], \ldots, \E[M^{(k)}] \right) \in \Qbar \left[ \E^{\leq \ell} \right] \mid p \left( \E[M^{(1)}_n], \ldots, \E[M^{(k)}_n] \right) = 0 \text{ for all } n \in \N_0 \right\}. 
\]
We refer to elements of $\mathbb{I}^{\leq \ell}$ as \emph{polynomial moment invariants}.
\end{definition}

Intuitively, the moment invariant ideal $\mathbb{I}^{\leq \ell}$ is the set of \emph{all} polynomials in the moments $\E^{\leq \ell}$ that vanish after any number of executed transitions.
For example, using Definition~\ref{def:mom:PInv}, a polynomial $p(\E[x], \E[y])$ in the expected values of the variables $x$ and $y$ is a \emph{polynomial moment invariant}, if $p(\E[x_n], \E[y_n]) = 0$ for all number of transitions $n \in \N_0$.
Note that, although $\E^{\leq \ell}$ is a finite set, the moment invariant ideal $\mathbb{I}^{\leq \ell}$ is, in general, an infinite set.

\begin{example}\label{ex:moment-inv-ideal}
Consider two asymmetric random walks $x_n$ and $y_n$ that both start at the origin.
Both random walks increase or decrease with probability $\nicefrac{1}{2}$, respectively.
The random walk $x_n$ either decreases by $2$ or increases by $1$, while $y_n$ behaves conversely, which means $y_n$ either decreases by $1$ or increases by $2$.
Following is a probabilistic loop encoding this process together with the moment invariant ideal $\mathbb{I}^{\leq 2}$.
The loop is given as program code.
The intended meaning of the expression $e_1 [pr] e_2$ is that it evaluates to $e_1$ with probability $pr$ and to $e_2$ with probability $1{-}pr$.
\medskip

\noindent
\begin{minipage}{\linewidth}
\begin{minipage}{0.35\linewidth}
\begin{algorithmic}
\State{$
    \begin{bmatrix}
        x~ & y
    \end{bmatrix}
    \gets
    \begin{bmatrix}
        0~ & 0
    \end{bmatrix}
    $
}
\While{$\star$}
    \vspace*{1mm}
    \State{$
        \begin{bmatrix}
            x \\ y
        \end{bmatrix}
        \gets
        \begin{bmatrix}
            x + 2 \ \left[ \nicefrac{1}{2} \right] \ x-1 \\ 
            y + 1 \ \left[ \nicefrac{1}{2} \right] \ y-2
        \end{bmatrix}
        $
    }
    \vspace*{1mm}
\EndWhile
\end{algorithmic}
\end{minipage}\hfill
\begin{minipage}{0.55\linewidth}
\problembox{
Basis of the moment invariant ideal $\I^{\leq 2}$:
\begin{align*}
    & \E \left[ x^2 \right] - \E \left[y^2 \right] \\
    & 9 \cdot \E[x] - 2 \cdot \E[xy] - 2 \cdot \E \left[y^2 \right] \\
    & \E[xy]^2 + 2 \cdot \E[xy] \cdot \E \left[y^2 \right] + \nicefrac{81}{4} \cdot \E[xy] + \E \left[y^2 \right]^2 \\
    & 2 \cdot \E[xy] + 9 \cdot \E[y] + 2 \cdot \E \left[y^2 \right] 
\end{align*}
}
\end{minipage}
\end{minipage}

\medskip
\noindent
This ideal $\mathbb{I}^{\leq 2}$ contains all algebraic relations that hold among $\E[x_n]$, $\E[y_n]$, $\E \left[ x_n^2 \right]$, $\E \left[ y_n^2 \right]$ and $\E[(xy)_n]$ after all number of iterations $n \in \N_0$.
The ideal provides information about the stochastic process encoded by the loop.
For instance, using the basis, it can be automatically checked that $\E[xy] - \E[x]\E[y]$ is an element of $\mathbb{I}^{\leq 2}$.
Hence, $\E[xy] = \E[x]\E[y]$ is an invariant, witnessing $x$ and $y$ being uncorrelated.
\end{example}

Moment invariant ideals of Definition~\ref{def:moment-ideal} generalize the notion of classical invariant ideals of Definition~\ref{def:invIdeal} for nonprobabilistic loops.
For a program variable $x$ of a nonprobabilistic loop, the expected value of $x$ after $n$ transitions is just the value of $x$ after $n$ iterations, that is $\E[x_n] = x_n$.
Furthermore, $\E[x_n \cdot y_n] = x_n \cdot y_n$ for all program variables $x$ and $y$.
Hence, a moment invariant such as $\E[x^2]^3 - \E[y]\E[y^2]$ corresponds to the classical invariant $x^6 - y^3$.
To formalize this observation, we introduce a function $\psi$ mapping invariants involving moments to classical invariants.

\begin{definition}[From Moment Invariants to Invariants]
Let $\mathcal{P}$ be a program with variables $x_1, \ldots, x_k$.
We define the natural \emph{ring homomorphism} $\psi\colon \overline{\Q}[\E^{\leq \ell}] \to \overline{\Q}[x_1, \ldots, x_k]$ extending $\psi(\E[M]) := M$.
That means, for all $p,q \in \overline{\Q}[\E^{\leq \ell}]$ and $c \in \overline{\Q}$ the function $\psi$ satisfies the properties (i) $\psi(p + q) = \psi(p) + \psi(q)$; (ii) $\psi(p \cdot q) = \psi(p) \cdot \psi(q)$; and (iii) $\psi(c \cdot p) = c \cdot \psi(p)$.
\end{definition}

The function $\psi$ maps polynomials over moments to polynomials over program variables, for example, $\psi(\E[x^2]^3 - \E[y]\E[y^2]) = \psi(\E[x^2])^3 - \psi(\E[y])\psi(\E[y^2]) = x^6 - y^3$.
If $p$ is a polynomial moment invariant of a \emph{probabilistic} program, $\psi(p)$ is in general \emph{not} a classical invariant.
However, for nonprobabilistic programs, $\psi(p)$ is necessarily an invariant for every moment invariant $p$, as we show in the next lemma.

\begin{lemma}[Moment Invariant Ideal Generalization]\label{lemma:moment-inv-generalization}
Let $\mathcal{L}$ be a \emph{nonprobabilistic} loop.
Let $\mathcal{I}$ be the classical invariant ideal and $\mathbb{I}^{\leq \ell}$ the moment invariant ideal of order $\ell$.
Then, $\mathbb{I}^{\leq \ell}$ and $\mathcal{I}$ are identical under $\psi$, that is 
\[
    \psi \left( \mathbb{I}^{\leq \ell} \right) := \left \{ \psi(p) \mid p \in \mathbb{I}^{\leq \ell} \right \} = \mathcal{I}.
\]
\end{lemma}
\begin{proof}
We show that $\psi(\mathbb{I}^{\leq \ell}) \subseteq \mathcal{I}$.
The reasoning for $\mathcal{I} \subseteq \psi(\mathbb{I}^{\leq \ell})$ is analogous.

Let $q \in \psi(\mathbb{I}^{\leq \ell})$.
Then, there is a $p(\E[M^{(1)}], \ldots, \E[M^{(m)}]) \in \mathbb{I}^{\leq \ell}$ for some monomials in program variables $M^{(i)}$ such that $\psi(p) = p(M^{(1)}, \ldots, M^{(m)}) = q$.
The polynomial $p$ in moments of program variables is an invariant because it is an element of $\mathbb{I}^{\leq \ell}$.
Moreover, because the loop $\mathcal{L}$ is nonprobabilistic, we have $\E[M_n] = M_n$ for all number of transitions $n \in \N_0$ and all monomials $M$ in program variables \footnote{If the loop contains nondeterministic choice, this property holds with respect to every scheduler resolving nondeterminism. For readability and simplicity, we omit the treatment of schedulers and refer to~\cite{Barthe2020} for details on schedulers.}.
Hence, $q = p(M^{(1)}, \ldots, M^{(m)})$ necessarily is a classical invariant as in Definition~\ref{def:inv} and therefore $q \in \mathcal{I}$.
\end{proof}

Lemma~\ref{lemma:moment-inv-generalization} hence proves that Definition~\ref{def:moment-ideal} generalizes the notion of invariant ideals of nonprobabilistic loops.

\subsection{Computability of Moment Invariant Ideals}\label{sec:spinf:comp}
We next consider a special class of probabilistic loops, called \emph{moment-computable polynomial loops}.
For such loops, we prove that the bases for moment invariant ideals $\mathbb{I}^{\leq \ell}$ are computable for any order $\ell$.
Moreover, in Algorithm~\ref{alg:moment-inv-ideal} we give a decision procedure computing moment invariant ideals of moment-computable polynomial loops. 

Let us recall the semantical notion of \emph{moment-computable loops}~\cite{polar}, which we adjusted to our setting of polynomial probabilistic loops.

\begin{definition}[Moment-Computable Polynomial Loops]\label{def:MomCompLoop}
A polynomial probabilistic loop $\mathcal{L}$ is \emph{moment-computable} if, for any monomial $M$ in loop variables of $\mathcal{L}$, we have that $\E[M_n]$ exists and is computable as $\E[M_n]=f(n)$, where $f(n)$ is an exponential polynomial in $n$, describing sums of polynomials multiplied by exponential terms in $n$. That is, $f(n)=\sum_{i=0}^k p_i(n) \cdot \lambda^n$ where all $p_i \in \overline{\Q}[n]$ are polynomials and $\lambda \in \overline{\Q}$.
\end{definition}

As stated in~\cite{kauers2011concrete}, we note that any \lrs~\eqref{eq:Cfinite} has an exponential polynomial as closed form.
As proven in~\cite{polar}, when considering loops with affine assignments, probabilistic choice with constant probabilities, and drawing from probability distributions with constant parameters and existing moments, all moments of program variables follow linear recurrence sequences.
Moreover, one may also consider polynomial (and not just affine) loop updates such that non-linear dependencies among variables are acyclic.
If-statements can also be supported if the loop guards contain only program variables with a finite domain.
Under such structural considerations, the resulting probabilistic loops are moment-computable loops~\cite{polar}:
expected values $\E[M_n]$ for monomials $M$ over loop variables are exponential polynomials in $n$.
Furthermore, a basis for the polynomial relations among exponential polynomials is computable~\cite{KauersZ08}.
We thus obtain a decision procedure computing the bases of moment invariant ideals of moment-computable polynomial loops, as given in Algorithm~\ref{alg:moment-inv-ideal} and discussed next.

\begin{algorithm}[t]
\caption{Computing moment invariant ideals}\label{alg:moment-inv-ideal}
\begin{algorithmic}
\Require A moment-computable polynomial loop $\mathcal{L}$ and an order $\ell \in \N$
\Ensure A basis $B$ for the moment invariant ideal $\mathbb{I}^{\leq \ell}$
\LeftComment{Closed forms of moments as exponential polynomials}
\State $C \gets \texttt{compute_closed_forms}(\mathcal{L}, \E^{\leq \ell})$
\LeftComment{A basis for the ideal of all algebraic relations among sequences in $C$}
\State $B \gets \texttt{compute_algebraic_relations}(C)$
\State \Return $B$
\end{algorithmic}
\end{algorithm}

The procedure $\texttt{compute_closed_form}(\mathcal{L}, S)$ in Algorithm~\ref{alg:moment-inv-ideal} takes as inputs a moment-computable polynomial loop $\mathcal{L}$ and a set $S$ of moments of loop variables and computes exponential polynomial closed forms of the moments in $S$;
here, we adjust results of~\cite{polar} to implement $\texttt{compute_closed_form}(\mathcal{L}, S)$.
Further, $\texttt{compute_algebraic_relations}(C)$ in Algorithm~\ref{alg:moment-inv-ideal} denotes a procedure that takes a set $C$ of exponential polynomial closed forms as input and computes a basis for all algebraic relations among them;
in our work, we use~\cite{KauersZ08} to implement $\texttt{compute_algebraic_relations}(C)$. 
Soundness of Algorithm~\ref{alg:moment-inv-ideal} follows from the soundness arguments of~\cite{polar,KauersZ08}.
We implemented Algorithm~\ref{alg:moment-inv-ideal} in our tool called \texttt{Polar}\footnote{\url{https://github.com/probing-lab/polar}}, allowing us to automatically derive the strongest polynomial moment invariants of moment-computable polynomial loops.

\begin{example}
Using Algorithm~\ref{alg:moment-inv-ideal} for the probabilistic loop of Example~\ref{ex:moment-inv-ideal}, we compute a basis for the moment invariant ideal $\mathbb{I}^{\leq 2}$ in approximately $0.4$ seconds and for $\mathbb{I}^{\leq 3}$ in roughly $0.8$ seconds, on a machine with a \SI{2.6}{GHz} Intel i7 processor and \SI{32}{GB} of RAM.
\end{example}

\subsection{Hardness for Guarded Probabilistic Loops}\label{sec:spinv:guarded}

As Algorithm~\ref{alg:moment-inv-ideal} provides a decision procedure for moment-computable polynomial loops, a natural question is whether the moment invariant ideals remain computable if we relax 
\begin{enumerate}[label={{\bf(C\arabic*)}},wide=0em,leftmargin=0em]
\item \label{C1} the restrictions on the guards,
\item \label{C2} the structural requirements on the polynomial assignments
\end{enumerate}
of moment-computable polynomial loops.

We first focus on~\ref{C1}, that is, lifting the restriction on guards and show that in this case a basis for the moment invariant ideal of any order becomes uncomputable (Theorem~\ref{thm:prob-invariants-pcp}).

We recall the seminal result of~\cite{Muller-OlmS04} proving that the strongest polynomial invariant for \emph{nonprobabilistic} loops with affine updates, nondeterministic choice, and guarded transitions is uncomputable.
Interestingly, nondeterministic choice can be replaced by uniform probabilistic choice, allowing us to also establish the uncomputability of the strongest polynomial moment invariants, which means a basis for the ideal $\mathbb{I}^{\leq \ell}$, for any order $\ell$.

\begin{theorem}[Uncomputability of Moment Invariant Ideal]\label{thm:prob-invariants-pcp}
For the class of guarded probabilistic loops with affine updates, a basis for the moment invariant ideal $\mathbb{I}^{\leq \ell}$ is uncomputable for any order $\ell$.
\end{theorem}
\begin{proof}
The proof is by reduction from Post's correspondence problem (PCP), which is undecidable \cite{Post46}.
A PCP instance consists of a finite alphabet $\Sigma$ and a finite set of tuples $\left \{ (x_i, y_i) \mid 1 \leq i \leq N, x_i, y_i \in {\Sigma}^\ast \right \}$.
A solution is a sequence of indices $(i_k)$, $1 \leq k \leq K$ where $i_k \in \{ 1, \ldots, N \}$ and the concatenations of the substrings indexed by the sequence are identical, written in symbols as
\[ 
    x_{i_1} \cdot x_{i_2} \cdot \ldots \cdot x_{i_K} = y_{i_1} \cdot y_{i_2} \cdot \ldots \cdot y_{i_K} 
\]
Note that the tuple elements may be of different lengths.
Moreover, any instance of the PCP over a finite alphabet $\Sigma$ can be equivalently represented over the alphabet $\{ 0,1 \}$ by a binary encoding.

Now, given an instance of the (binary) PCP, we construct the guarded probabilistic loop with affine updates shown in Figure~\ref{fig:prob-invariants:pcp}.
We encode the binary strings as integers and denote a transition with probability $pr$, guard $g$ and updates $f$ as $[pr]:g:\vec{x}\gets f(\vec{x})$.

\smallskip
\begin{figure}[htb]
    \centering
    \begin{tikzpicture}[->,>=stealth',shorten >=1pt,auto,node distance=5cm, semithick, outer sep=auto] 
\tikzset{
dots/.style={state,draw=none}, 
edge/.style={->}
}
\tikzstyle{every state}=[fill=white,draw=black,thick,text=black,scale=1, initial text={}] 
\path[use as bounding box] (-0.5,-0.5) rectangle
        (13,1.5);
\node[state, initial above] (q0) {$q_0$};
\node[state] (q1) [right of=q0] {$q_1$}; 
\draw[->] (q1) to [in=60, out=120, loop, looseness=6] node[above] {$[1] : (x = y \land x > 0): t \gets 1$} (q1);
\draw[->] (q0) to [in=180, out=0] node[above] {$[1]:\top: x, y, t \gets 0, 0, 0$} (q1);
\draw[->] (q1) to [in=-45, out=45, loop, looseness=6] node[right] {
    \begin{tabular}{l}
         for each $1 \leq i \leq N$:  \\[0.1cm]
         $[\nicefrac{1}{N}] : (x \neq y \lor x \leq 0): 
            \begin{bmatrix}
                x \\ y
            \end{bmatrix}
            \gets
            \begin{bmatrix}
                2^{\lvert x_i \rvert} \cdot x + x_i \\
                2^{\lvert y_i \rvert} \cdot y + y_i
            \end{bmatrix}
            $ 
    \end{tabular}
} (q1);
\draw[->,dashed] (q1) to [in=-40, out=40, loop, looseness=5] (q1);
\draw[->,dashed] (q1) to [in=-35, out=35, loop, looseness=4] (q1);
\end{tikzpicture}
    \caption{A guarded probabilistic loop with affine updates simulating the PCP.}
    \label{fig:prob-invariants:pcp}
\end{figure}
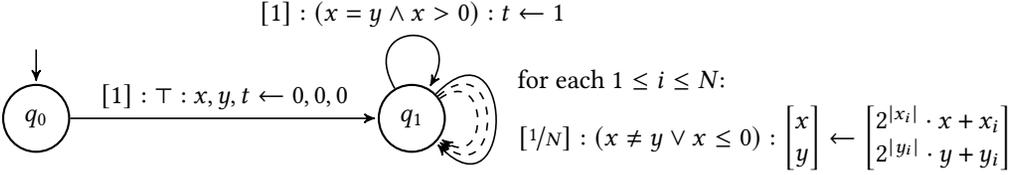

The idea is to pick a pair of integer-encoded strings uniformly at random and append them to the string built so far.
This is done by left-shifting the existing bits of the string (by multiplying by a power of $2$) and adding the randomly selected string.

If the PCP instance does not have a solution, we have $t = 0$ after every transition.
Hence, $\E[t] = 0$ must be an invariant.
Therefore, $\E[t]$ is necessarily an element of $\mathbb{I}^{\leq \ell}$ for any order $\ell$.

If the PCP instance does have a solution $(i_k), 1 \leq k \leq K$, then after exactly $n=K+2$ transitions it holds that $\P(x_n = y_n) \geq \left( \frac{1}{N} \right)^K$, as this is the probability of choosing the correct sequence uniformly at random. 
Because $t$ is an indicator variable, $\E[t_n] = \P(t_n = 1) = \P(x_n = y_n) \geq \left( \frac{1}{N} \right)^K > 0$.
Hence, $\E[t_n] \neq 0$ after $n$ transitions and $\E[t]$ cannot be an element of $\mathbb{I}^{\leq \ell}$ for any order $\ell$.

Consequently, for all orders $\ell$, the PCP instance has a solution if and only if $\E[t]$ is an element of $\mathbb{I}^{\leq \ell}$.
However, given a basis, checking for ideal membership is decidable (cf. Section~\ref{sec:preliminaries:alg-geom}).
Hence, a basis for the moment invariant ideal $\mathbb{I}^{\leq \ell}$ must be uncomputable for any order $\ell$.
\end{proof}

Note that the PCP reduction within the proof of Theorem~\ref{thm:prob-invariants-pcp} requires only affine updates and affine invariants.
Therefore, allowing loop guards renders even the problem of finding the strongest affine invariant for a finite set of moments uncomputable for probabilistic loops with affine updates.

\subsection{Hardness for Unguarded Polynomial Probabilistic Loops}\label{sec:spinv:unguarded}

In this section we address challenge~\ref{C2}, that is, study computational lower bounds for computing a basis of moment invariant ideals for probabilistic loops that lack guards and nondeterminism, but feature arbitrary polynomial updates.
We show that addressing \ref{C2} boils down to solving the \probspinv problem of Section~\ref{sec:intro}, which in turn we prove to be \skolem-hard (Theorem~\ref{thm:spinv-probspinv}).
As such, computing the moment invariant ideals of probabilistic loops with arbitrary polynomial updates as stated in~\ref{C2} is \skolem-hard. 

We restrict our attention to moment invariant ideals of order $1$.
Intuitively, a basis for $\mathbb{I}^{\leq 1}$ is easier to compute than $\mathbb{I}^{\leq \ell}$ for $\ell > 1$.
A  formal justification in this respect is given by the following lemma.

\begin{lemma}[Moment Invariant Ideal of Order 1]\label{lem:basis1}
Given a basis for the moment invariant ideal $\mathbb{I}^{\leq \ell}$ for any order $\ell \in \N$, a basis for $\mathbb{I}^{\leq 1}$ is computable.
\end{lemma}
\begin{proof}
The moment invariant ideal $\mathbb{I}^{\leq \ell}$ is an ideal in the polynomial ring with variables $\E^{\leq \ell}$.
Moreover, $\E^{\leq 1} \subseteq \E^{\leq \ell}$.
Hence, $\mathbb{I}^{\leq 1} = \mathbb{I}^{\leq \ell} \cap \overline{\Q}[\E^{\leq 1}]$, meaning $\mathbb{I}^{\leq 1}$ is an elimination ideal of $\mathbb{I}^{\leq \ell}$.
Given a basis for a polynomial ideal, bases for elimination ideals are computable \cite{CoxLittleOshea97}.
\end{proof}

Using Lemma~\ref{lem:basis1}, we translate challenge~\ref{C2} into the \probspinv problem of Section~\ref{sec:intro}, formally defined as follows. 

\problembox{The \probspinv Problem: Given an unguarded, probabilistic loop with polynomial updates and without nondeterministic choice, compute a basis of the moment invariant ideal of order $1$.}

Recall that computing a basis for the classical invariant ideal for nonprobabilistic programs with arbitrary polynomial updates, that is, deciding \spinv, is \skolem-hard (Theorem~\ref{theorem:skolem} and Theorem~\ref{theorem:ptp-spinv}). 
We next show that \spinv reduces to \probspinv, thus implying \skolem-hardness of \probspinv as a direct consequence of Lemma~\ref{lemma:moment-inv-generalization}.

\begin{theorem}[Hardness of \probspinv]\label{thm:spinv-probspinv}
\probspinv is at least as hard as \spinv, in symbols $\spinv \leq \probspinv$.
\end{theorem}
\begin{proof}
Assume $\mathcal{L}$ is an instance of \spinv.
That is, $\mathcal{L}$ is a deterministic loop with polynomial updates.
Let $x_1, \ldots, x_k$ be the program variables and $\mathcal{I}$ the classical invariant ideal of $\mathcal{L}$.
Note that $\mathcal{L}$ is also an instance of \probspinv and assume $B$ is a basis for the moment invariant ideal $\mathcal{\I}^{\leq 1}$.
From Lemma~\ref{lemma:moment-inv-generalization} we know that $\psi(\mathcal{\I}^{\leq 1}) = \mathcal{I}$.
For order $1$, the function $\psi$ is a ring isomorphism between the polynomial rings $\overline{\Q}[x_1, \ldots, x_k]$ and $\overline{\Q}[\E[x_1], \ldots, \E[x_k]]$.
Hence, the set $\{ \psi(b) \mid b \in B \}$ is a basis for $\mathcal{I}$.
Therefore, given a basis for $\mathbb{I}^{\leq 1}$, a basis for $\mathcal{I}$ is computable.
\end{proof}

Theorem~\ref{thm:spinv-probspinv} shows that \probspinv is at least as hard as the \spinv problem. 
Together with Theorem~\ref{theorem:skolem} and Theorem~\ref{theorem:ptp-spinv}, we conclude the following chain of reductions: 
$$\skolem \leq \ptp \leq \spinv \leq \probspinv$$

\paragraph{On attempting to prove uncomputability of \probspinv -- A remaining open challenge.}
While Theorem~\ref{thm:spinv-probspinv} asserts that \probspinv is \skolem-hard, it could be that \probspinv is uncomputable.

Recall that for proving the uncomputability of moment invariant ideals for guarded probabilistic programs in Theorem~\ref{thm:prob-invariants-pcp}, we replaced nondeterministic choice with probabilistic choice.
The \lq\lq nondeterministic version\rq\rq\ of \probspinv refers to computing the strongest polynomial invariant for nondeterministic polynomial programs, which has been recently established as uncomputable~\cite{Ouaknine23}.
Therefore, it is natural to consider transferring the uncomputability results of~\cite{Ouaknine23} to \probspinv by replacing nondeterministic choice with probabilistic choice.
However, such a generalization of~\cite{Ouaknine23} to the probabilistic setting poses considerable problems and ultimately fails to establish the potential uncomputability of \probspinv, for the reasons discussed next.

The proof in \cite{Ouaknine23} reduces the Boundedness problem for Reset Vector Addition System with State (VASS) to the problem of finding the strongest polynomial invariant for nondeterministic polynomial programs.
A Reset VASS is a nondeterministic program where any transition may increment, decrement, or reset a vector of unbounded, non-negative variables.
Importantly, a transition can \emph{only be executed if no zero-valued variable is decremented}.
The \emph{Boundedness Problem for Reset VASS} asks, given a Reset VASS and a specific program location, whether the set of reachable program configurations is finite. The Boundedness Problem for Reset VASS is undecidable \cite{DufourdFS98} and therefore instrumental in the reduction of~\cite{Ouaknine23}.

Namely, in the reduction of~\cite{Ouaknine23} to prove uncomputability of the strongest polynomial invariant for nondeterministic polynomial programs, an arbitrary Reset VASS $\mathcal{V}$ with $n$ variables $a_1, \ldots, a_n$ is simulated by a nondeterministic polynomial program $\mathcal{P}$ with $n{+}1$ variables $b_0, \ldots b_n$.
Note that the programming model is purely nondeterministic, that is, without equality guards, since introducing guards would render the problem immediately undecidable~\cite{Muller-OlmS04}.
To avoid zero-testing the variables before executing a transition, the crucial point in the reduction of~\cite{Ouaknine23} is to map invalid traces to the vector $\vec{0}$ and faithfully simulate valid executions.
By properties of the reduction, it holds that the configuration $(b_0, \ldots, b_n)$ is reachable in $\mathcal{P}$, if and only if there exists a corresponding configuration $\nicefrac{1}{b_0} \cdot (b_1, \ldots, b_n)$ in $\mathcal{V}$.
Essential to the reduction of~\cite{Ouaknine23} is, that even though there may be multiple configurations in $\mathcal{P}$ for each configuration in $\mathcal{V}$, all these configurations are only scaled by the factor $b_0$ and hence collinear. 
By collinearity, the variety of the invariant ideal can be covered by a finite set of lines if and only if the set of reachable VASS configurations is finite. Testing this property is decidable, and hence finding the invariant ideal must be undecidable.

Transferring the reduction of~\cite{Ouaknine23} to the probabilistic setting of \probspinv by replacing nondeterministic choice with probabilistic choice poses the following problem:
in the nondeterministic setting, any path is independent of all other paths.
However, this does not hold in the probabilistic setting of \probspinv.  
The expected value operator $\E[x_n]$ aggregates all possible valuations of $x$ in iteration $n$ across all possible paths through the program.
Specifically, the expected value is a linear combination of the possible configurations of $\mathcal{V}$, which is not necessarily limited to a collection of lines but may span a higher-dimensional subspace.
This is the step where a reduction similar to~\cite{Ouaknine23} fails for \probspinv.

\begin{example}\label{ex:prob-invariants-proof}
Consider a Reset VASS $\mathcal{V}$ with variable $x$ initialized to $0$, initial state $q_0$, and additional state $q_1$.
Assume a single transition from $q_0$ to $q_1$ incrementing $x$ and two transitions from $q_1$ to $q_1$.
One transition from $q_1$ to $q_1$ decrements $x$, whereas the other leaves $x$ unchanged.
In a Reset VASS, it is forbidden to decrement a zero-valued variable.
Therefore, the set of reachable configurations in $q_1$ is $\{0, 1\}$ and hence finite.
The reduction in \cite{Ouaknine23} constructs from $\mathcal{V}$ a nondeterministic polynomial program $\mathcal{P}$ with two variables $y$ and $z$.
Similar to $\mathcal{V}$, the program $\mathcal{P}$ has two states $\hat{q_0}$ and $\hat{q_1}$, one transition from $\hat{q_0}$ to $\hat{q_1}$ and two transitions from $\hat{q_1}$ to itself.
In contrast to $\mathcal{V}$, the transitions in $\mathcal{P}$ model polynomial assignments for the variables $y$ and $z$.
For more details on the reduction, we refer to \cite{Ouaknine23}.
Important are the reachable configurations of $\mathcal{P}$ depicted in the computation tree in Figure~\ref{fig:prob-invariants:comp-tree}.
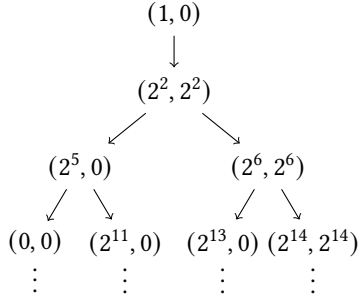
\begin{figure}[htb]
    \centering
    \begin{tikzpicture}[->] 
\tikzset{
level 2/.style = {sibling distance = 2.5cm},
level 3/.style = {sibling distance = 1.2cm}
}

\node {$(1,0)$} [level distance = 1cm]
    child {node {$(2^2,2^2)$}
        child {node {$(2^5,0)$}
            child {node {$(0,0)$}[level distance = 0.5cm]
                child {edge from parent[draw=none] node {\vdots}}
            }
            child {node {$(2^{11},0)$}[level distance = 0.5cm]
                child {edge from parent[draw=none] node {\vdots}}
            } 
        }
        child {node {$(2^6,2^6)$}
            child {node {$(2^{13},0)$}[level distance = 0.5cm]
                child {edge from parent[draw=none] node {\vdots}}
            }
            child {node {$(2^{14},2^{14})$}[level distance = 0.5cm]
                child {edge from parent[draw=none] node {\vdots}}
            } 
        }
    };

\end{tikzpicture}
    \caption{Computation tree of the program $\mathcal{P}$ from Example~\ref{ex:prob-invariants-proof}.}
    \label{fig:prob-invariants:comp-tree}
\end{figure}
For every reachable configuration $(y,z) \neq (0,0)$ we have $\nicefrac{z}{y} \in \{0,1\}$
Hence, all reachable configurations lie on finitely many lines.
Replacing nondeterministic choice in the state $\hat{q_1}$ by uniform probabilistic choice and considering expected values breaks this central property of the reduction.
The sequence of expected values for $y$ and $z$ can be obtained by averaging the variable values for every level in the computation tree in Figure~\ref{fig:prob-invariants:comp-tree} and is $(1,0), (4,4), (48, 32), (4096, 6656), \dots$.
It can be calculated that the ratios of the expected values after $n \geq 1$ transitions are given by
$$
\frac{\E[z_n]}{\E[y_n]} = \frac{1}{\sum_{i=0}^{n-1} \frac{1}{2^{2^i - 1}}},
$$
and hence the points $\{ (\E[y_n], \E[z_n]) \mid n \in \N \}$ cannot be covered with finitely many lines.
\end{example}

It is however worth noting how well-suited the Boundedness Problem for Reset VASS is for proving the undecidability of problems for unguarded programs.
A Reset VASS is not powerful enough to determine if a variable is zero, yet the Boundedness Problem is still undecidable.
The vast majority of other undecidable problems that may be used in a reduction are formulated in terms of counter-machines, Turing machines, or other automata that rely on explicitly determining if a given variable is zero, hindering a straightforward simulation as unguarded programs.
Therefore, we conjecture that any attempt towards proving (un)computability of \probspinv would require a new methodology, unrelated to~\cite{Ouaknine23}. We leave this task as an open challenge for future work.

\subsection{Summary of Computability Results for Probabilistic Polynomial Loop Invariants}

We finally conclude this section by summarizing our computability results on the strongest polynomial (moment) invariants of probabilistic loops. We overview our results in Table~\ref{strongest-invariant:table:overview-prob-invariants}.

\begin{table}[htb]
    \setlength{\tabcolsep}{0.5em}
    \centering
    \resizebox{\textwidth}{!}{%
    \begin{tabular}{|l|l|l||cl|cl|} \hline
        \multicolumn{3}{|l||}{Program Model} & \multicolumn{2}{c}{Strongest Affine Invariant} & \multicolumn{2}{|c|}{Strongest Polynomial Invariant} \\ \hline
        \multirow{4}*{Prob.} & Unguarded \& & Affine & \checkmark & Algorithm~\ref{alg:moment-inv-ideal} & \checkmark & Algorithm~\ref{alg:moment-inv-ideal} \\ \cline{3-7}
        & Guarded (finite) & Poly. & \textbf{?} & & \skolem-hard & Theorem~\ref{thm:spinv-probspinv} \\ \cline{2-7}
        & \multirow{2}*{Guarded ($=, <$)} & Affine & \multicolumn{4}{c|}{\multirow{2}{*}{\xmark \hspace{0.2cm} Theorem~\ref{thm:prob-invariants-pcp}}} \\ \cline{3-3}
        & & Poly. & \multicolumn{4}{c|}{} \\ \hline
    \end{tabular}}
    \medskip
    \caption{Our computability results for strongest polynomial (moment) invariants of polynomial \emph{probabilistic} loops. The symbol '\checkmark' denotes computable problems, '\textbf{?}' shows open problems, and '\xmark' marks uncomputable problems.}
    \label{strongest-invariant:table:overview-prob-invariants}
\end{table}

\section{Related Work}\label{sec:related-work}

We discuss our work in relation to the state-of-the-art in computing strongest (probabilistic) invariants and analyzing point-to-point reachability.

\paragraph{Strongest Invariants.}
Algebraic invariants were first considered for unguarded deterministic programs with affine updates \cite{Karr76}.
Here, a basis for both the ideal of affine invariants and for the ideal of polynomial invariants is computable \cite{Karr76, Kovacs08}.

For unguarded deterministic programs with polynomial updates, all invariants of \emph{bounded degree} are computable \cite{Muller-OlmS04-2}, while the more general task of computing a basis for the ideal of \emph{all} polynomial invariants, that is solving our \spinv problem, was stated as an open problem.
In Section~\ref{sec:ptp-spinv} we proved that \spinv is at least as hard as \skolem and \ptp.
Strengthening these results by proving computability for \spinv would result in a major breakthrough in number theory, as this would imply the decidability of the \skolem problem.

For guarded deterministic programs, the strongest affine invariant is uncomputable, even for programs with only affine updates. 
This is a direct consequence of the fact that this model is sufficient to encode Turing machines and allows us to encode the Halting problem \cite{HopcroftUllman69}.
Nevertheless, there exists a multitude of incomplete methods capable of extracting useful invariants even for non-linear programs, for example, based on abstract domains~\cite{KincaidCBR18}, over-approximation in combination with recurrences \cite{FarzanK15,KincaidBCR19} or using consequence finding in tractable logical theories of non-linear arithmetic \cite{KincaidKZ23}.

For nondeterministic programs with affine updates, a basis for the invariant ideal is computable \cite{Karr76}.
Furthermore, the set of invariants of bounded degree is computable for nondeterministic programs with polynomial updates, while bases for the ideal of all invariants are uncomputable \cite{Muller-OlmS04-2,Ouaknine23}.
Additionally, even a single transition guarded by an equality or inequality predicate renders the problem uncomputable, already for affine updates \cite{Muller-OlmS04}.

\paragraph{Point-To-Point Reachability.}
The Point-To-Point reachability problem formalized by our \ptp problem appears in various areas dealing with discrete systems, such as dynamical systems, discrete mathematics, and program analysis.
For linear dynamical systems, \ptp is known as the \emph{Orbit problem}~\cite{ChonevOW13}, with a significant amount of work on analyzing and proving decidability of \ptp for linear systems \cite{KannanL80,ChonevOW13,ChonevOW15,Baier0JKLLOPW021}.
In contrast, for polynomial systems, the \ptp problem remained open regarding decidability or computational lower bounds.
Existing techniques in this respect resorted to approximate techniques \cite{Dreossi17,DangT12}.
Contrarily to these works, in Section~\ref{sec:skolem-ptp} we rigorously proved that \ptp for polynomial systems is at least as hard as the \skolem problem.
The \ptp problem is essentially undecidable already for affine systems that additionally include nondeterministic choice \cite{Finkel13,Ko18}.

\paragraph{Probabilistic Invariants.}
Invariants for probabilistic loops can be defined in various incomparable ways, depending on the context and use case.
Dijkstra's weakest-precondition calculus for classical programs was generalized to the weakest-preexpectation (wp) calculus in the seminal works~\cite{Kozen83,Kozen85,McIver05}.
In the wp-calculus, the semantics of a loop can be described as the least fixed point of the \emph{characteristic} function of the loop in the lattice of so-called \emph{expectations}~\cite{Kaminski19}.
Invariants are expectations that over- or under-approximate this fixed point and are called super- or sub-invariants, respectively.
One line of research is to synthesize such invariants using templates and constraint-solving methods \cite{Gretz13,BatzCKKMS20,BatzCJKKM23}.
A calculus, analogous to the wp-calculus, has been introduced for expected runtime analysis \cite{KaminskiKMO18} and amortized expected runtime analysis \cite{BatzKKMV23}.
The work of~\cite{ChatterjeeNZ17} introduces the notion of \emph{stochastic invariants}, that is, expressions that are violated with bounded probability.
Other notions of probabilistic invariants involve martingale theory~\cite{BartheEFH16} or utilize bounds on the expected value of program variable expressions \cite{Chakarov14}.
The techniques presented in \cite{polar,Bartocci19} compute closed forms for moments of program variables parameterized by the loop counter.

The different notions of probabilistic invariants, in general, do not form ideals or are relative to some other expression. 
Furthermore, the existing procedures to compute invariants are heuristics-driven and hence incomplete. Contrarily to these, our \emph{polynomial moment invariants} presented in Section~\ref{sec:prob-invariants} form ideals and relate all variables.
Moreover, our Algorithm~\ref{alg:moment-inv-ideal} computes a basis for \emph{all} moment invariants and is complete for the class of moment-computable polynomial loops. Going beyond such loops, we showed that \probspinv is \skolem-hard and/or uncomputable (Theorem~\ref{thm:spinv-probspinv} and Theorem~\ref{thm:prob-invariants-pcp}).

\section{Conclusion}\label{sec:conclusion}

We prove that computing the strongest polynomial invariant for single-path loops with polynomial assignments (\spinv) is at least as hard as the \skolem problem, a famous problem whose decidability has been open for almost a century.
As such, we provide the first non-trivial lower bound for computing the strongest polynomial invariant for deterministic polynomial loops, a quest introduced in \cite{Muller-OlmS04}.
As an intermediate result, we show that point-to-point reachability in deterministic polynomial loops (\ptp), or equivalently in discrete-time polynomial dynamical systems, is \skolem-hard. Further, we devise a reduction from \ptp to \spinv.
We generalize the notion of invariant ideals from classical programs to the probabilistic setting, by introducing \emph{moment invariant ideals} and addressing the \probspinv problem.
We show that the strongest polynomial moment invariant, and hence \probspinv, is (i) computable for the class of \emph{moment-computable} probabilistic loops, but becomes (ii) uncomputable for probabilistic loops with branching statements and (iii) \skolem-hard for polynomial probabilistic loops without branching statements. Going beyond \skolem-hardness of \probspinv and \spinv are open challenges we aim to further study.

\begin{acks}
This research was supported by 
the European Research Council Consolidator Grant ARTIST 101002685, the 
Vienna Science and Technology Fund WWTF 10.47379/ICT19018 grant ProbInG,  and the SecInt Doctoral College funded by TU Wien.
We thank Manuel Kauers for providing details on sequences and algebraic relations and Toghrul Karimov for inspiring us to consider the orbit problem.
We thank the McGill Bellairs Research Institute for hosting the Bellairs 2023 workshop, whose fruitful discussions influenced parts of this work.
\end{acks}

\bibliographystyle{ACM-Reference-Format}
\bibliography{references}


\begin{thebibliography}{53}


\ifx \showCODEN    \undefined \def \showCODEN     #1{\unskip}     \fi
\ifx \showDOI      \undefined \def \showDOI       #1{#1}\fi
\ifx \showISBNx    \undefined \def \showISBNx     #1{\unskip}     \fi
\ifx \showISBNxiii \undefined \def \showISBNxiii  #1{\unskip}     \fi
\ifx \showISSN     \undefined \def \showISSN      #1{\unskip}     \fi
\ifx \showLCCN     \undefined \def \showLCCN      #1{\unskip}     \fi
\ifx \shownote     \undefined \def \shownote      #1{#1}          \fi
\ifx \showarticletitle \undefined \def \showarticletitle #1{#1}   \fi
\ifx \showURL      \undefined \def \showURL       {\relax}        \fi
\providecommand\bibfield[2]{#2}
\providecommand\bibinfo[2]{#2}
\providecommand\natexlab[1]{#1}
\providecommand\showeprint[2][]{arXiv:#2}

\bibitem[Baier et~al\mbox{.}(2021)]%
        {Baier0JKLLOPW021}
\bibfield{author}{\bibinfo{person}{Christel Baier}, \bibinfo{person}{Florian
  Funke}, \bibinfo{person}{Simon Jantsch}, \bibinfo{person}{Toghrul Karimov},
  \bibinfo{person}{Engel Lefaucheux}, \bibinfo{person}{Florian Luca},
  \bibinfo{person}{Jo{\"{e}}l Ouaknine}, \bibinfo{person}{David Purser},
  \bibinfo{person}{Markus~A. Whiteland}, {and} \bibinfo{person}{James
  Worrell}.} \bibinfo{year}{2021}\natexlab{}.
\newblock \showarticletitle{The Orbit Problem for Parametric Linear Dynamical
  Systems}. In \bibinfo{booktitle}{\emph{Proc. of {CONCUR}}}.
\newblock
\urldef\tempurl%
\url{https://doi.org/10.4230/LIPIcs.CONCUR.2021.28}
\showDOI{\tempurl}


\bibitem[Barthe et~al\mbox{.}(2016)]%
        {BartheEFH16}
\bibfield{author}{\bibinfo{person}{Gilles Barthe}, \bibinfo{person}{Thomas
  Espitau}, \bibinfo{person}{Luis Mar{\'{\i}}a~Ferrer Fioriti}, {and}
  \bibinfo{person}{Justin Hsu}.} \bibinfo{year}{2016}\natexlab{}.
\newblock \showarticletitle{Synthesizing Probabilistic Invariants via Doob's
  Decomposition}. In \bibinfo{booktitle}{\emph{Proc. of {CAV}}}.
\newblock
\urldef\tempurl%
\url{https://doi.org/10.1007/978-3-319-41528-4\_3}
\showDOI{\tempurl}


\bibitem[Barthe et~al\mbox{.}(2012a)]%
        {Barthe2012}
\bibfield{author}{\bibinfo{person}{Gilles Barthe}, \bibinfo{person}{Benjamin
  Gr{\'{e}}goire}, {and} \bibinfo{person}{Santiago~Zanella B{\'{e}}guelin}.}
  \bibinfo{year}{2012}\natexlab{a}.
\newblock \showarticletitle{Probabilistic Relational Hoare Logics for
  Computer-Aided Security Proofs}. In \bibinfo{booktitle}{\emph{Proc. of
  {MPC}}}.
\newblock
\urldef\tempurl%
\url{https://doi.org/10.1007/978-3-642-31113-0}
\showDOI{\tempurl}


\bibitem[Barthe et~al\mbox{.}(2020)]%
        {Barthe2020}
\bibfield{author}{\bibinfo{person}{Gilles Barthe},
  \bibinfo{person}{Joost-Pieter Katoen}, {and} \bibinfo{person}{Alexandra
  Silva}.} \bibinfo{year}{2020}\natexlab{}.
\newblock \bibinfo{booktitle}{\emph{Foundations of Probabilistic Programming}}.
\newblock \bibinfo{publisher}{Cambridge University Press}.
\newblock
\urldef\tempurl%
\url{https://doi.org/10.1017/9781108770750}
\showDOI{\tempurl}


\bibitem[Barthe et~al\mbox{.}(2012b)]%
        {Barthe2012a}
\bibfield{author}{\bibinfo{person}{Gilles Barthe}, \bibinfo{person}{Boris
  K{\"{o}}pf}, \bibinfo{person}{Federico Olmedo}, {and}
  \bibinfo{person}{Santiago~Zanella B{\'{e}}guelin}.}
  \bibinfo{year}{2012}\natexlab{b}.
\newblock \showarticletitle{Probabilistic Relational Reasoning for Differential
  Privacy}. In \bibinfo{booktitle}{\emph{Proc. of {POPL}}}.
\newblock
\urldef\tempurl%
\url{https://doi.org/10.1145/2103656.2103670}
\showDOI{\tempurl}


\bibitem[Bartocci et~al\mbox{.}(2019)]%
        {Bartocci19}
\bibfield{author}{\bibinfo{person}{Ezio Bartocci}, \bibinfo{person}{Laura
  Kov{\'{a}}cs}, {and} \bibinfo{person}{Miroslav Stankovic}.}
  \bibinfo{year}{2019}\natexlab{}.
\newblock \showarticletitle{Automatic Generation of Moment-Based Invariants for
  Prob-Solvable Loops}. In \bibinfo{booktitle}{\emph{Proc. of {ATVA}}}.
\newblock
\urldef\tempurl%
\url{https://doi.org/10.1007/978-3-030-31784-3\_15}
\showDOI{\tempurl}


\bibitem[Batz et~al\mbox{.}(2023a)]%
        {BatzCJKKM23}
\bibfield{author}{\bibinfo{person}{Kevin Batz}, \bibinfo{person}{Mingshuai
  Chen}, \bibinfo{person}{Sebastian Junges}, \bibinfo{person}{Benjamin~Lucien
  Kaminski}, \bibinfo{person}{Joost{-}Pieter Katoen}, {and}
  \bibinfo{person}{Christoph Matheja}.} \bibinfo{year}{2023}\natexlab{a}.
\newblock \showarticletitle{Probabilistic Program Verification via Inductive
  Synthesis of Inductive Invariants}. In \bibinfo{booktitle}{\emph{Proc. of
  {TACAS}}}.
\newblock
\urldef\tempurl%
\url{https://doi.org/10.1007/978-3-031-30820-8\_25}
\showDOI{\tempurl}


\bibitem[Batz et~al\mbox{.}(2021)]%
        {BatzCKKMS20}
\bibfield{author}{\bibinfo{person}{Kevin Batz}, \bibinfo{person}{Mingshuai
  Chen}, \bibinfo{person}{Benjamin~Lucien Kaminski},
  \bibinfo{person}{Joost{-}Pieter Katoen}, \bibinfo{person}{Christoph Matheja},
  {and} \bibinfo{person}{Philipp Schr{\"{o}}er}.}
  \bibinfo{year}{2021}\natexlab{}.
\newblock \showarticletitle{Latticed k-Induction with an Application to
  Probabilistic Programs}. In \bibinfo{booktitle}{\emph{Proc. of {CAV}}}.
\newblock
\urldef\tempurl%
\url{https://doi.org/10.1007/978-3-030-81688-9\_25}
\showDOI{\tempurl}


\bibitem[Batz et~al\mbox{.}(2023b)]%
        {BatzKKMV23}
\bibfield{author}{\bibinfo{person}{Kevin Batz},
  \bibinfo{person}{Benjamin~Lucien Kaminski}, \bibinfo{person}{Joost{-}Pieter
  Katoen}, \bibinfo{person}{Christoph Matheja}, {and} \bibinfo{person}{Lena
  Verscht}.} \bibinfo{year}{2023}\natexlab{b}.
\newblock \showarticletitle{A Calculus for Amortized Expected Runtimes}.
\newblock \bibinfo{journal}{\emph{Proc. {ACM} Program. Lang.}}
  \bibinfo{number}{{POPL}} (\bibinfo{year}{2023}).
\newblock
\urldef\tempurl%
\url{https://doi.org/10.1145/3571260}
\showDOI{\tempurl}


\bibitem[Bilu et~al\mbox{.}(2022)]%
        {Bilu22}
\bibfield{author}{\bibinfo{person}{Yuri Bilu}, \bibinfo{person}{Florian Luca},
  \bibinfo{person}{Joris Nieuwveld}, \bibinfo{person}{Jo{\"{e}}l Ouaknine},
  \bibinfo{person}{David Purser}, {and} \bibinfo{person}{James Worrell}.}
  \bibinfo{year}{2022}\natexlab{}.
\newblock \showarticletitle{Skolem Meets Schanuel}. In
  \bibinfo{booktitle}{\emph{Proc. of {MFCS}}}.
\newblock
\urldef\tempurl%
\url{https://doi.org/10.4230/LIPIcs.MFCS.2022.20}
\showDOI{\tempurl}


\bibitem[Buchberger(2006)]%
        {Buchberger-thesis}
\bibfield{author}{\bibinfo{person}{Bruno Buchberger}.}
  \bibinfo{year}{2006}\natexlab{}.
\newblock \showarticletitle{Bruno Buchberger's PhD thesis 1965: An algorithm
  for finding the basis elements of the residue class ring of a zero
  dimensional polynomial ideal}.
\newblock \bibinfo{journal}{\emph{J. Symb. Comput.}} (\bibinfo{year}{2006}).
\newblock
\urldef\tempurl%
\url{https://doi.org/10.1016/j.jsc.2005.09.007}
\showDOI{\tempurl}


\bibitem[Cadilhac et~al\mbox{.}(2020)]%
        {CadilhacMPPS20}
\bibfield{author}{\bibinfo{person}{Micha{\"{e}}l Cadilhac},
  \bibinfo{person}{Filip Mazowiecki}, \bibinfo{person}{Charles Paperman},
  \bibinfo{person}{Michal Pilipczuk}, {and} \bibinfo{person}{G{\'{e}}raud
  S{\'{e}}nizergues}.} \bibinfo{year}{2020}\natexlab{}.
\newblock \showarticletitle{On Polynomial Recursive Sequences}. In
  \bibinfo{booktitle}{\emph{Proc. of {ICALP}}}.
\newblock
\urldef\tempurl%
\url{https://doi.org/10.4230/LIPIcs.ICALP.2020.117}
\showDOI{\tempurl}


\bibitem[Chakarov and Sankaranarayanan(2014)]%
        {Chakarov14}
\bibfield{author}{\bibinfo{person}{Aleksandar Chakarov} {and}
  \bibinfo{person}{Sriram Sankaranarayanan}.} \bibinfo{year}{2014}\natexlab{}.
\newblock \showarticletitle{Expectation Invariants for Probabilistic Program
  Loops as Fixed Points}. In \bibinfo{booktitle}{\emph{Proc. of {SAS}}}.
\newblock
\urldef\tempurl%
\url{https://doi.org/10.1007/978-3-319-10936-7\_6}
\showDOI{\tempurl}


\bibitem[Chatterjee et~al\mbox{.}(2017)]%
        {ChatterjeeNZ17}
\bibfield{author}{\bibinfo{person}{Krishnendu Chatterjee},
  \bibinfo{person}{Petr Novotn{\'{y}}}, {and} \bibinfo{person}{Dorde Zikelic}.}
  \bibinfo{year}{2017}\natexlab{}.
\newblock \showarticletitle{Stochastic invariants for probabilistic
  termination}. In \bibinfo{booktitle}{\emph{Proc. of {POPL}}}.
\newblock
\urldef\tempurl%
\url{https://doi.org/10.1145/3009837.3009873}
\showDOI{\tempurl}


\bibitem[Chonev et~al\mbox{.}(2013)]%
        {ChonevOW13}
\bibfield{author}{\bibinfo{person}{Ventsislav Chonev},
  \bibinfo{person}{Jo{\"{e}}l Ouaknine}, {and} \bibinfo{person}{James
  Worrell}.} \bibinfo{year}{2013}\natexlab{}.
\newblock \showarticletitle{The orbit problem in higher dimensions}. In
  \bibinfo{booktitle}{\emph{Proc. of {STOC}}}.
\newblock
\urldef\tempurl%
\url{https://doi.org/10.1145/2488608.2488728}
\showDOI{\tempurl}


\bibitem[Chonev et~al\mbox{.}(2015)]%
        {ChonevOW15}
\bibfield{author}{\bibinfo{person}{Ventsislav Chonev},
  \bibinfo{person}{Jo{\"{e}}l Ouaknine}, {and} \bibinfo{person}{James
  Worrell}.} \bibinfo{year}{2015}\natexlab{}.
\newblock \showarticletitle{The Polyhedron-Hitting Problem}. In
  \bibinfo{booktitle}{\emph{Proc. of {SODA}}}.
\newblock
\urldef\tempurl%
\url{https://doi.org/10.1137/1.9781611973730.64}
\showDOI{\tempurl}


\bibitem[Cox et~al\mbox{.}(1997)]%
        {CoxLittleOshea97}
\bibfield{author}{\bibinfo{person}{David~A. Cox}, \bibinfo{person}{John
  Little}, {and} \bibinfo{person}{Donal O'Shea}.}
  \bibinfo{year}{1997}\natexlab{}.
\newblock \bibinfo{booktitle}{\emph{Ideals, varieties, and algorithms - an
  introduction to computational algebraic geometry and commutative algebra}}.
\newblock
\urldef\tempurl%
\url{https://doi.org/10.1137/1035171}
\showDOI{\tempurl}


\bibitem[Dang and Testylier(2012)]%
        {DangT12}
\bibfield{author}{\bibinfo{person}{Thao Dang} {and} \bibinfo{person}{Romain
  Testylier}.} \bibinfo{year}{2012}\natexlab{}.
\newblock \showarticletitle{Reachability Analysis for Polynomial Dynamical
  Systems Using the Bernstein Expansion}.
\newblock \bibinfo{journal}{\emph{Reliab. Comput.}} (\bibinfo{year}{2012}).
\newblock


\bibitem[Dreossi et~al\mbox{.}(2017)]%
        {Dreossi17}
\bibfield{author}{\bibinfo{person}{Tommaso Dreossi}, \bibinfo{person}{Thao
  Dang}, {and} \bibinfo{person}{Carla Piazza}.}
  \bibinfo{year}{2017}\natexlab{}.
\newblock \showarticletitle{Reachability computation for polynomial dynamical
  systems}.
\newblock \bibinfo{journal}{\emph{Formal Methods Syst. Des.}}
  (\bibinfo{year}{2017}).
\newblock
\urldef\tempurl%
\url{https://doi.org/10.1007/s10703-016-0266-3}
\showDOI{\tempurl}


\bibitem[Dufourd et~al\mbox{.}(1998)]%
        {DufourdFS98}
\bibfield{author}{\bibinfo{person}{Catherine Dufourd}, \bibinfo{person}{Alain
  Finkel}, {and} \bibinfo{person}{Philippe Schnoebelen}.}
  \bibinfo{year}{1998}\natexlab{}.
\newblock \showarticletitle{Reset Nets Between Decidability and
  Undecidability}. In \bibinfo{booktitle}{\emph{Proc. of {ICALP}}}.
\newblock
\urldef\tempurl%
\url{https://doi.org/10.1007/BFb0055044}
\showDOI{\tempurl}


\bibitem[Everest et~al\mbox{.}(2003)]%
        {Everest-Skolem}
\bibfield{author}{\bibinfo{person}{Graham Everest}, \bibinfo{person}{Alfred~J.
  van~der Poorten}, \bibinfo{person}{Igor~E. Shparlinski}, {and}
  \bibinfo{person}{Thomas Ward}.} \bibinfo{year}{2003}\natexlab{}.
\newblock \bibinfo{booktitle}{\emph{Recurrence Sequences}}.
\newblock \bibinfo{publisher}{American Mathematical Society}.
\newblock
\showISBNx{978-0-8218-3387-2}
\newblock
\shownote{ISBN 978-0-8218-3387-2}.


\bibitem[Farzan and Kincaid(2015)]%
        {FarzanK15}
\bibfield{author}{\bibinfo{person}{Azadeh Farzan} {and}
  \bibinfo{person}{Zachary Kincaid}.} \bibinfo{year}{2015}\natexlab{}.
\newblock \showarticletitle{Compositional Recurrence Analysis}. In
  \bibinfo{booktitle}{\emph{{FMCAD}}}.
\newblock
\urldef\tempurl%
\url{https://doi.org/10.1109/FMCAD.2015.7542253}
\showDOI{\tempurl}


\bibitem[Finkel et~al\mbox{.}(2013)]%
        {Finkel13}
\bibfield{author}{\bibinfo{person}{Alain Finkel}, \bibinfo{person}{Stefan
  G{\"{o}}ller}, {and} \bibinfo{person}{Christoph Haase}.}
  \bibinfo{year}{2013}\natexlab{}.
\newblock \showarticletitle{Reachability in Register Machines with Polynomial
  Updates}. In \bibinfo{booktitle}{\emph{Proc. of {MFCS}}}.
\newblock
\urldef\tempurl%
\url{https://doi.org/10.1007/978-3-642-40313-2\_37}
\showDOI{\tempurl}


\bibitem[Ghahramani(2015)]%
        {Ghahramani2015}
\bibfield{author}{\bibinfo{person}{Zoubin Ghahramani}.}
  \bibinfo{year}{2015}\natexlab{}.
\newblock \showarticletitle{Probabilistic Machine Learning and Artificial
  Intelligence}.
\newblock \bibinfo{journal}{\emph{Nature}} (\bibinfo{year}{2015}).
\newblock
\urldef\tempurl%
\url{https://doi.org/10.1038/nature14541}
\showDOI{\tempurl}


\bibitem[Gretz et~al\mbox{.}(2013)]%
        {Gretz13}
\bibfield{author}{\bibinfo{person}{Friedrich Gretz},
  \bibinfo{person}{Joost{-}Pieter Katoen}, {and} \bibinfo{person}{Annabelle
  McIver}.} \bibinfo{year}{2013}\natexlab{}.
\newblock \showarticletitle{Prinsys - On a Quest for Probabilistic Loop
  Invariants}. In \bibinfo{booktitle}{\emph{Proc. of {QEST}}}.
\newblock
\urldef\tempurl%
\url{https://doi.org/10.1007/978-3-642-40196-1\_17}
\showDOI{\tempurl}


\bibitem[Hopcroft and Ullman(1969)]%
        {HopcroftUllman69}
\bibfield{author}{\bibinfo{person}{John~E. Hopcroft} {and}
  \bibinfo{person}{Jeffrey~D. Ullman}.} \bibinfo{year}{1969}\natexlab{}.
\newblock \bibinfo{booktitle}{\emph{Formal languages and their relation to
  automata}}.
\newblock


\bibitem[Hrushovski et~al\mbox{.}(2018)]%
        {DBLP:conf/lics/HrushovskiOP018}
\bibfield{author}{\bibinfo{person}{Ehud Hrushovski},
  \bibinfo{person}{Jo{\"{e}}l Ouaknine}, \bibinfo{person}{Amaury Pouly}, {and}
  \bibinfo{person}{James Worrell}.} \bibinfo{year}{2018}\natexlab{}.
\newblock \showarticletitle{{Polynomial Invariants for Affine Programs}}. In
  \bibinfo{booktitle}{\emph{Proc. of {LICS}}}.
\newblock
\urldef\tempurl%
\url{https://doi.org/10.1145/3209108.3209142}
\showDOI{\tempurl}


\bibitem[Hrushovski et~al\mbox{.}(2023)]%
        {Ouaknine23}
\bibfield{author}{\bibinfo{person}{Ehud Hrushovski}, \bibinfo{person}{Jo\"{e}l
  Ouaknine}, \bibinfo{person}{Amaury Pouly}, {and} \bibinfo{person}{James
  Worrell}.} \bibinfo{year}{2023}\natexlab{}.
\newblock \showarticletitle{On Strongest Algebraic Program Invariants}.
\newblock \bibinfo{journal}{\emph{J. ACM}} (\bibinfo{year}{2023}).
\newblock
\urldef\tempurl%
\url{https://doi.org/10.1145/3614319}
\showDOI{\tempurl}


\bibitem[Kaminski et~al\mbox{.}(2019)]%
        {Kaminski19}
\bibfield{author}{\bibinfo{person}{Benjamin~Lucien Kaminski},
  \bibinfo{person}{Joost{-}Pieter Katoen}, {and} \bibinfo{person}{Christoph
  Matheja}.} \bibinfo{year}{2019}\natexlab{}.
\newblock \showarticletitle{On the hardness of analyzing probabilistic
  programs}.
\newblock \bibinfo{journal}{\emph{Acta Inform.}} (\bibinfo{year}{2019}).
\newblock
\urldef\tempurl%
\url{https://doi.org/10.1007/s00236-018-0321-1}
\showDOI{\tempurl}


\bibitem[Kaminski et~al\mbox{.}(2018)]%
        {KaminskiKMO18}
\bibfield{author}{\bibinfo{person}{Benjamin~Lucien Kaminski},
  \bibinfo{person}{Joost{-}Pieter Katoen}, \bibinfo{person}{Christoph Matheja},
  {and} \bibinfo{person}{Federico Olmedo}.} \bibinfo{year}{2018}\natexlab{}.
\newblock \showarticletitle{Weakest Precondition Reasoning for Expected
  Runtimes of Randomized Algorithms}.
\newblock \bibinfo{journal}{\emph{J. {ACM}}} (\bibinfo{year}{2018}).
\newblock
\urldef\tempurl%
\url{https://doi.org/10.1145/3208102}
\showDOI{\tempurl}


\bibitem[Kannan and Lipton(1980)]%
        {KannanL80}
\bibfield{author}{\bibinfo{person}{Ravindran Kannan} {and}
  \bibinfo{person}{Richard~J. Lipton}.} \bibinfo{year}{1980}\natexlab{}.
\newblock \showarticletitle{The Orbit Problem is Decidable}. In
  \bibinfo{booktitle}{\emph{Proc. of {STOC}}}.
\newblock
\urldef\tempurl%
\url{https://doi.org/10.1145/800141.804673}
\showDOI{\tempurl}


\bibitem[Karimov et~al\mbox{.}(2022)]%
        {Karimov22}
\bibfield{author}{\bibinfo{person}{Toghrul Karimov}, \bibinfo{person}{Engel
  Lefaucheux}, \bibinfo{person}{Jo{\"{e}}l Ouaknine}, \bibinfo{person}{David
  Purser}, \bibinfo{person}{Anton Varonka}, \bibinfo{person}{Markus~A.
  Whiteland}, {and} \bibinfo{person}{James Worrell}.}
  \bibinfo{year}{2022}\natexlab{}.
\newblock \showarticletitle{What's decidable about linear loops?}
\newblock \bibinfo{journal}{\emph{Proc. {ACM} Program. Lang.}}
  \bibinfo{number}{{POPL}} (\bibinfo{year}{2022}).
\newblock
\urldef\tempurl%
\url{https://doi.org/10.1145/3498727}
\showDOI{\tempurl}


\bibitem[Karr(1976)]%
        {Karr76}
\bibfield{author}{\bibinfo{person}{Michael Karr}.}
  \bibinfo{year}{1976}\natexlab{}.
\newblock \showarticletitle{Affine Relationships Among Variables of a Program}.
\newblock \bibinfo{journal}{\emph{Acta Inform.}} (\bibinfo{year}{1976}).
\newblock
\urldef\tempurl%
\url{https://doi.org/10.1007/BF00268497}
\showDOI{\tempurl}


\bibitem[Kauers(2005)]%
        {kauers-thesis}
\bibfield{author}{\bibinfo{person}{Manuel Kauers}.}
  \bibinfo{year}{2005}\natexlab{}.
\newblock \emph{\bibinfo{title}{Algorithms for Nonlinear Higher Order
  Difference Equations}}.
\newblock \bibinfo{thesistype}{Ph.\,D. Dissertation}. \bibinfo{school}{RISC,
  Johannes Kepler University, Linz}.
\newblock


\bibitem[Kauers and Paule(2011)]%
        {kauers2011concrete}
\bibfield{author}{\bibinfo{person}{Manuel Kauers} {and} \bibinfo{person}{Peter
  Paule}.} \bibinfo{year}{2011}\natexlab{}.
\newblock \bibinfo{booktitle}{\emph{{The Concrete Tetrahedron - Symbolic Sums,
  Recurrence Equations, Generating Functions, Asymptotic Estimates}}}.
\newblock \bibinfo{publisher}{Springer}.
\newblock
\urldef\tempurl%
\url{https://doi.org/10.1007/978-3-7091-0445-3}
\showDOI{\tempurl}


\bibitem[Kauers and Zimmermann(2008)]%
        {KauersZ08}
\bibfield{author}{\bibinfo{person}{Manuel Kauers} {and}
  \bibinfo{person}{Burkhard Zimmermann}.} \bibinfo{year}{2008}\natexlab{}.
\newblock \showarticletitle{Computing the algebraic relations of C-finite
  sequences and multisequences}.
\newblock \bibinfo{journal}{\emph{J. Symb. Comput.}} (\bibinfo{year}{2008}).
\newblock
\urldef\tempurl%
\url{https://doi.org/10.1016/J.JSC.2008.03.002}
\showDOI{\tempurl}


\bibitem[Kincaid et~al\mbox{.}(2019)]%
        {KincaidBCR19}
\bibfield{author}{\bibinfo{person}{Zachary Kincaid}, \bibinfo{person}{Jason
  Breck}, \bibinfo{person}{John Cyphert}, {and} \bibinfo{person}{Thomas~W.
  Reps}.} \bibinfo{year}{2019}\natexlab{}.
\newblock \showarticletitle{Closed forms for numerical loops}.
\newblock \bibinfo{journal}{\emph{Proc. {ACM} Program. Lang.}}
  \bibinfo{number}{{POPL}} (\bibinfo{year}{2019}).
\newblock
\urldef\tempurl%
\url{https://doi.org/10.1145/3290368}
\showDOI{\tempurl}


\bibitem[Kincaid et~al\mbox{.}(2018)]%
        {KincaidCBR18}
\bibfield{author}{\bibinfo{person}{Zachary Kincaid}, \bibinfo{person}{John
  Cyphert}, \bibinfo{person}{Jason Breck}, {and} \bibinfo{person}{Thomas~W.
  Reps}.} \bibinfo{year}{2018}\natexlab{}.
\newblock \showarticletitle{Non-linear reasoning for invariant synthesis}.
\newblock \bibinfo{journal}{\emph{Proc. {ACM} Program. Lang.}}
  \bibinfo{number}{{POPL}} (\bibinfo{year}{2018}).
\newblock
\urldef\tempurl%
\url{https://doi.org/10.1145/3158142}
\showDOI{\tempurl}


\bibitem[Kincaid et~al\mbox{.}(2023)]%
        {KincaidKZ23}
\bibfield{author}{\bibinfo{person}{Zachary Kincaid}, \bibinfo{person}{Nicolas
  Koh}, {and} \bibinfo{person}{Shaowei Zhu}.} \bibinfo{year}{2023}\natexlab{}.
\newblock \showarticletitle{When Less Is More: Consequence-Finding in a Weak
  Theory of Arithmetic}.
\newblock \bibinfo{journal}{\emph{Proc. {ACM} Program. Lang.}}
  \bibinfo{number}{{POPL}} (\bibinfo{year}{2023}).
\newblock
\urldef\tempurl%
\url{https://doi.org/10.1145/3571237}
\showDOI{\tempurl}


\bibitem[Ko et~al\mbox{.}(2018)]%
        {Ko18}
\bibfield{author}{\bibinfo{person}{Sang{-}Ki Ko}, \bibinfo{person}{Reino
  Niskanen}, {and} \bibinfo{person}{Igor Potapov}.}
  \bibinfo{year}{2018}\natexlab{}.
\newblock \showarticletitle{Reachability Problems in Nondeterministic
  Polynomial Maps on the Integers}. In \bibinfo{booktitle}{\emph{Proc. of
  {DLT}}}.
\newblock
\urldef\tempurl%
\url{https://doi.org/10.1007/978-3-319-98654-8\_38}
\showDOI{\tempurl}


\bibitem[Kofnov et~al\mbox{.}(2022)]%
        {KofnovMSBB22}
\bibfield{author}{\bibinfo{person}{Andrey Kofnov}, \bibinfo{person}{Marcel
  Moosbrugger}, \bibinfo{person}{Miroslav Stankovic}, \bibinfo{person}{Ezio
  Bartocci}, {and} \bibinfo{person}{Efstathia Bura}.}
  \bibinfo{year}{2022}\natexlab{}.
\newblock \showarticletitle{Moment-Based Invariants for Probabilistic Loops
  with Non-polynomial Assignments}. In \bibinfo{booktitle}{\emph{Proc. of
  {QEST}}}.
\newblock
\urldef\tempurl%
\url{https://doi.org/10.1007/978-3-031-16336-4\_1}
\showDOI{\tempurl}


\bibitem[Kov{\'{a}}cs(2008)]%
        {Kovacs08}
\bibfield{author}{\bibinfo{person}{Laura Kov{\'{a}}cs}.}
  \bibinfo{year}{2008}\natexlab{}.
\newblock \showarticletitle{{Reasoning Algebraically About P-Solvable Loops}}.
  In \bibinfo{booktitle}{\emph{Proc. of {TACAS}}}.
\newblock
\urldef\tempurl%
\url{https://doi.org/10.1007/978-3-540-78800-3\_18}
\showDOI{\tempurl}


\bibitem[Kov{\'{a}}cs and Varonka(2023)]%
        {Varonka23}
\bibfield{author}{\bibinfo{person}{Laura Kov{\'{a}}cs} {and}
  \bibinfo{person}{Anton Varonka}.} \bibinfo{year}{2023}\natexlab{}.
\newblock \showarticletitle{What Else is Undecidable About Loops?}. In
  \bibinfo{booktitle}{\emph{Proc. of {RAMiCS}}}.
\newblock
\urldef\tempurl%
\url{https://doi.org/10.1007/978-3-031-28083-2\_11}
\showDOI{\tempurl}


\bibitem[Kozen(1983)]%
        {Kozen83}
\bibfield{author}{\bibinfo{person}{Dexter Kozen}.}
  \bibinfo{year}{1983}\natexlab{}.
\newblock \showarticletitle{A Probabilistic {PDL}}. In
  \bibinfo{booktitle}{\emph{Proc. of {STOC}}}.
\newblock
\urldef\tempurl%
\url{https://doi.org/10.1145/800061.808758}
\showDOI{\tempurl}


\bibitem[Kozen(1985)]%
        {Kozen85}
\bibfield{author}{\bibinfo{person}{Dexter Kozen}.}
  \bibinfo{year}{1985}\natexlab{}.
\newblock \showarticletitle{A Probabilistic {PDL}}.
\newblock \bibinfo{journal}{\emph{J. Comput. Syst. Sci.}}
  (\bibinfo{year}{1985}).
\newblock
\urldef\tempurl%
\url{https://doi.org/10.1016/0022-0000(85)90012-1}
\showDOI{\tempurl}


\bibitem[Lipton et~al\mbox{.}(2022)]%
        {Lipton22}
\bibfield{author}{\bibinfo{person}{Richard Lipton}, \bibinfo{person}{Florian
  Luca}, \bibinfo{person}{Joris Nieuwveld}, \bibinfo{person}{Jo{\"{e}}l
  Ouaknine}, \bibinfo{person}{David Purser}, {and} \bibinfo{person}{James
  Worrell}.} \bibinfo{year}{2022}\natexlab{}.
\newblock \showarticletitle{On the Skolem Problem and the Skolem Conjecture}.
  In \bibinfo{booktitle}{\emph{Proc. of {LICS}}}.
\newblock
\urldef\tempurl%
\url{https://doi.org/10.1145/3531130.3533328}
\showDOI{\tempurl}


\bibitem[McIver and Morgan(2005)]%
        {McIver05}
\bibfield{author}{\bibinfo{person}{Annabelle McIver} {and}
  \bibinfo{person}{Carroll Morgan}.} \bibinfo{year}{2005}\natexlab{}.
\newblock \bibinfo{booktitle}{\emph{Abstraction, Refinement and Proof for
  Probabilistic Systems}}.
\newblock
\urldef\tempurl%
\url{https://doi.org/10.1007/b138392}
\showDOI{\tempurl}


\bibitem[Moosbrugger et~al\mbox{.}(2022)]%
        {polar}
\bibfield{author}{\bibinfo{person}{Marcel Moosbrugger},
  \bibinfo{person}{Miroslav Stankovic}, \bibinfo{person}{Ezio Bartocci}, {and}
  \bibinfo{person}{Laura Kov{\'{a}}cs}.} \bibinfo{year}{2022}\natexlab{}.
\newblock \showarticletitle{This is the moment for probabilistic loops}.
\newblock \bibinfo{journal}{\emph{Proc. {ACM} Program. Lang.}}
  \bibinfo{number}{{OOPSLA2}} (\bibinfo{year}{2022}).
\newblock
\urldef\tempurl%
\url{https://doi.org/10.1145/3563341}
\showDOI{\tempurl}


\bibitem[M{\"{u}}ller{-}Olm and Seidl(2004a)]%
        {Muller-OlmS04-2}
\bibfield{author}{\bibinfo{person}{Markus M{\"{u}}ller{-}Olm} {and}
  \bibinfo{person}{Helmut Seidl}.} \bibinfo{year}{2004}\natexlab{a}.
\newblock \showarticletitle{Computing polynomial program invariants}.
\newblock \bibinfo{journal}{\emph{Inf. Process. Lett.}} (\bibinfo{year}{2004}).
\newblock
\urldef\tempurl%
\url{https://doi.org/10.1016/j.ipl.2004.05.004}
\showDOI{\tempurl}


\bibitem[M{\"{u}}ller{-}Olm and Seidl(2004b)]%
        {Muller-OlmS04}
\bibfield{author}{\bibinfo{person}{Markus M{\"{u}}ller{-}Olm} {and}
  \bibinfo{person}{Helmut Seidl}.} \bibinfo{year}{2004}\natexlab{b}.
\newblock \showarticletitle{A Note on Karr's Algorithm}. In
  \bibinfo{booktitle}{\emph{Proc. of {ICALP}}}.
\newblock
\urldef\tempurl%
\url{https://doi.org/10.1007/978-3-540-27836-8\_85}
\showDOI{\tempurl}


\bibitem[M{\"{u}}llner(2023)]%
        {julians-thesis}
\bibfield{author}{\bibinfo{person}{Julian M{\"{u}}llner}.}
  \bibinfo{year}{2023}\natexlab{}.
\newblock \emph{\bibinfo{title}{Exact Inference for Probabilistic Loops}}.
\newblock \bibinfo{thesistype}{Master's\ thesis}. \bibinfo{school}{Technische
  Universit{\"{a}}t Wien}.
\newblock


\bibitem[Post(1946)]%
        {Post46}
\bibfield{author}{\bibinfo{person}{Emil~L. Post}.}
  \bibinfo{year}{1946}\natexlab{}.
\newblock \showarticletitle{{A variant of a recursively unsolvable problem}}.
\newblock \bibinfo{journal}{\emph{Bull. Am. Math. Soc.}}
  (\bibinfo{year}{1946}).
\newblock


\bibitem[Tao(2008)]%
        {Tao-Skolem}
\bibfield{author}{\bibinfo{person}{Terrence Tao}.}
  \bibinfo{year}{2008}\natexlab{}.
\newblock \bibinfo{booktitle}{\emph{Structure and Randomness}}.
\newblock \bibinfo{publisher}{American Mathematical Society}.
\newblock
\newblock
\shownote{ISBN 0-8218-4695-7}.


\end{thebibliography}

\end{document}